\documentclass[10pt,journal,compsoc]{IEEEtran}

\usepackage{amsmath,amsfonts}
\usepackage{amsthm}
\usepackage[ruled,vlined,linesnumbered]{algorithm2e}
\usepackage{array}
\usepackage{booktabs}
\usepackage{graphicx}
\usepackage{multirow}
\usepackage[table]{xcolor}
\definecolor{keyrow}{RGB}{190,214,250}
\definecolor{headrow}{gray}{0.75}
\definecolor{keyyellow}{RGB}{255,241,140}
\newcommand{\hb}[1]{\tikz[baseline=-0.55ex]{\draw[line width=0.3pt] (0,0) circle (0.75ex); \fill (0,0) -- (90:0.75ex) arc [start angle=90, end angle={90-360*#1}, radius=0.75ex] -- cycle;}}
\usepackage{etoolbox}
\AtBeginEnvironment{tabular}{\rowcolors{2}{gray!10}{white}}
\usepackage{url}
\usepackage{tikz}
\usepackage{listings}
\usepackage[hidelinks]{hyperref}
\usetikzlibrary{arrows.meta,positioning,calc}
\newtheorem{proposition}{Proposition}
\newtheorem{theorem}{Theorem}
\newtheorem{corollary}{Corollary}

\newcommand{\method}{FedMark-FM}
\newcommand{\valuator}{S3Val}
\newcommand{\bench}{FedMark-FM-Bench}
\begin{document}
\raggedbottom

\title{FedMark-FM: Auditable, Risk-Adjusted Data Markets for Federated Foundation-Model Adaptation}

\author{Phat~T.~Tran-Truong\textsuperscript{1,2,\,*},~
Xuan-Bach~Le\textsuperscript{1,2,\,\dag},~and~
Minh~Nhat~Nguyen\textsuperscript{3}%
\thanks{\textsuperscript{1}Faculty of Computer Science and Engineering, Ho Chi Minh City University of Technology (HCMUT), 268 Ly Thuong Kiet Street, Dien Hong Ward, Ho Chi Minh City, Vietnam.}%
\thanks{\textsuperscript{2}Vietnam National University Ho Chi Minh City, Linh Xuan Ward, Ho Chi Minh City, Vietnam.}%
\thanks{\textsuperscript{3}RMIT University, Ho Chi Minh City, Vietnam.}%
\thanks{E-mail: phatttt@hcmut.edu.vn; lexuanbach@hcmut.edu.vn; minh.nguyen244@rmit.edu.vn.}%
\thanks{\textsuperscript{*}First author.\quad\textsuperscript{\dag}Corresponding author.}}

\markboth{IEEE Transactions on Knowledge and Data Engineering}{Tran-Truong \MakeLowercase{\textit{et al.}}: FedMark-FM: Auditable, Risk-Adjusted Data Markets for Federated Foundation-Model Adaptation}

\maketitle

\begin{abstract}
Federated foundation-model adaptation increasingly relies on heterogeneous private artifacts---retrieval corpora, prompts and demonstrations, LoRA adapters, preference and safety data, and update sketches---yet existing federated-learning incentive mechanisms price clients as homogeneous data or update providers. This assumption is poorly matched to foundation-model pipelines, where contribution value is heterogeneous, non-IID, pipeline-dependent, privacy-constrained, and vulnerable to strategic behavior. We propose \textbf{\method{}} (a \textbf{Fed}erated \textbf{Mark}et for \textbf{F}oundation \textbf{M}odels), an auditable, risk-adjusted data-market framework that models clients as sellers of typed artifacts, estimates marginal contribution with \textbf{\valuator{}} (\textbf{S}ecure \textbf{S}urrogate \textbf{S}hapley \textbf{Val}uation)---a stratified, uncertainty-aware Shapley estimator that also supports pipeline-ordered valuation---and converts lower-confidence-bound values into budget-feasible payments that penalize duplication, sybil splitting, poisoned adapters, privacy-budget gaming, and cost inflation. We evaluate \textbf{\bench{}} (the \method{} \textbf{Bench}mark) across FEVER retrieval, held-out generator-backed RAG, and trained PEFT/LoRA tracks. Under a held-out prompt-injection poisoner, \method{} improves downstream accuracy by 7.5--8.1 points over volume, leave-one-out, and FL-Shapley while selecting zero strategic clients. Split-conformal calibration reaches full lower-bound coverage at mean width 0.0141, versus 0.33 for naive intervals. We prove that pipeline-ordered valuation is the unique credit rule respecting serving causality, and show it materially changes credit assignment (Spearman 0.76, selected-set overlap 0.67) while leaving held-out task quality unchanged; the market also preserves rare specialists with audit-ready ledgers at 200--1000-client scale. \method{} shows that incentives for federated foundation models can be engineered as auditable data infrastructure that couples valuation, mechanism design, privacy interfaces, and pipeline-order semantics.
\end{abstract}


\section{Introduction}

Foundation models are increasingly deployed as data pipelines rather than monolithic networks. A production assistant, search stack, or enterprise copilot answers a query by composing retrieved passages, prompt templates, in-context demonstrations, parameter-efficient adapters, preference data, and safety probes, and these artifacts are refreshed continuously after deployment. In cross-silo settings they are owned by different organizations---hospitals, data vendors, safety labs---that cannot or will not centralize raw data. Realizing federated foundation models (FedFMs) at this scale is therefore as much an economic problem as an algorithmic one: a shared pipeline improves only if the owners of these private artifacts are paid for their contribution, and they participate only if that contribution is measured and rewarded fairly, at scale, and with evidence they can dispute.

Two research lines bear on this problem. Federated-learning (FL) incentive mechanisms allocate rewards for participation and risk sharing, typically through Shapley-based or auction-based schemes over datasets or model updates \cite{kairouz2021advancesfl,yang2023flshapleypareto}. Data-valuation methods, led by Data Shapley and its efficient approximations, quantify each source's marginal contribution to a trained model \cite{ghorbani2019datashapley,jia2019efficient}. Recent FedFM surveys, in turn, raise incentives, game mechanisms, privacy, and heterogeneity to first-order open problems \cite{fan2025fedfmchallenges}.

However, these approaches price a client as a homogeneous provider of a dataset or gradient under a fixed supervised objective, and stop at a scalar reward computed after training. A deployable FedFM market needs four properties that no existing method provides together: heterogeneous artifacts must be traded as \emph{typed}, separately priced products; credit must respect the \emph{serving order}---retrieval precedes prompting, which precedes adaptation, which precedes safety---rather than conflating upstream and downstream contributions; valuation must \emph{scale under privacy constraints} that forbid inspecting raw artifacts; and payments must leave an \emph{auditable ledger} that can be disputed and revalued when the model, retriever, or policy changes.

In this paper we propose \method{}, an auditable, risk-adjusted data-market framework for FedFM adaptation that supplies all four properties. A client contributes typed artifacts---retrieval corpora, LoRA adapters, prompts, demonstrations, preference data, safety data, or update sketches---and the market evaluates coalitions through a contract utility spanning task quality, safety, robustness, latency, privacy budget, and cost. Its central conceptual move is to stop treating FedFM artifacts as exchangeable players: because retrieval changes the context seen by prompts and adapters, \method{} assigns credit along the contract-visible serving order, and we prove that this pipeline-ordered rule is the \emph{unique} credit assignment consistent with serving causality (Theorem~\ref{thm:ordered-char}).

Because exhaustive coalition evaluation is infeasible, \method{} estimates value with \valuator{} (Secure Surrogate Shapley Valuation), a stratified, uncertainty-aware Shapley estimator that combines contribution sketches, redundancy clustering, pipeline-ordered sampling, and a learned utility surrogate, and converts lower-confidence-bound values into budget-feasible payments that penalize duplication, sybil splitting, poisoned adapters, privacy-budget gaming, and cost inflation. On real FEVER retrieval, held-out generator-backed RAG, and trained PEFT/LoRA tracks, \method{} improves downstream accuracy by 7.5--8.1 points over volume, leave-one-out, and FL-Shapley under a held-out prompt-injection poisoner while selecting zero strategic clients, and its split-conformal payments attain full lower-bound coverage at mean interval width 0.0141.

To make the gap concrete, suppose a hospital contributes a small but unique biomedical retrieval corpus, a web vendor contributes a large but redundant one, a strategic client copies the hospital corpus under another identity, and a safety lab contributes probes that improve policy compliance without improving ordinary retrieval. Volume-based rewards overpay the vendor and the duplicate; a naive validation-loss delta underpays the hospital when biomedical queries are rare; and a standard Shapley estimator becomes too expensive once each coalition requires adapter routing or RAG evaluation. \method{} instead values typed artifacts under a contract utility, stratifies evaluations by domain and artifact type, discounts duplicate clusters, rewards scarce safety contributions, and audits high-risk payments, while retaining enough evidence to contest a valuation without exposing raw records. The framework is domain-agnostic; a licensed clinical or enterprise case study is future work rather than a claim of the present paper.

This work makes five contributions:

\begin{itemize}
  \item \textbf{Typed artifact market.} A FedFM data market that treats heterogeneous artifacts as first-class contribution units (Section~\ref{sec:formulation}).
  \item \textbf{Pipeline-ordered valuation.} A serving-aware credit rule that we prove is the \emph{unique} value satisfying standard axioms plus a serving-causal downstream-realization axiom, routing each coalition's complementarity dividend to its serving frontier (Theorem~\ref{thm:ordered-char}; Section~\ref{sec:formulation}).
  \item \textbf{Scalable valuation (\valuator{}).} A scalable, privacy-constrained valuation algorithm combining contribution sketches, stratified coalition sampling, surrogate modeling, and uncertainty-triggered audits (Section~\ref{sec:s3val}).
  \item \textbf{Risk-adjusted payments.} A budget-feasible payment rule on lower-confidence-bound values with penalties for cost, privacy, duplication, and manipulation risk (Section~\ref{sec:payment}).
  \item \textbf{Benchmark and system.}\footnote{Source code, benchmark harness, data, and reproduction scripts: \url{https://anonymous.4open.science/r/FedMark-FM-3A89}.} A deployed-style architecture, benchmark, and stress-test suite covering sybil, duplicate, poison, non-IID-suppression, privacy-gaming, and cost-inflation attacks (Sections~\ref{sec:system}, \ref{sec:bench}, \ref{sec:eval}).
\end{itemize}

\textbf{Pipeline-ordered credit.} \method{} supports both unordered \valuator{} and an ordered variant that restricts marginal-credit permutations to the contract-visible serving precedence, so a retrieval client is credited for improving downstream context while an adapter is credited conditional on the retrieval and prompt state that precedes it at serving time. Empirically, ordered valuation changes selected sets on FEVER (Spearman 0.7554, selected-set overlap 0.6667; risk-adjusted utility 0.2410 unordered vs.\ 0.2425 ordered). A controlled held-out serving study on real FEVER data (supplementary material,\footnote{The supplementary material---appendices covering registry and ledger schemas, notation, extended cost analysis, differential-privacy and attestation details, threat models, and full experimental protocols---is available with this submission. Its sections, tables, and figures are numbered with an ``S'' prefix.} Table~S15) shows that this redistribution is \emph{serving-neutral}: it changes which clients are credited so that payment respects serving causality, while leaving held-out task accuracy statistically unchanged (\(\Delta=+0.004\pm0.014\)). Ordered valuation is therefore a payment-fairness choice rather than an accuracy lever or an implementation detail.

\textbf{Significance.} The contribution is fundamentally a data-management and mechanism-design one that matters wherever private contributions are priced under uncertainty. A market round begins with a registry that binds each private artifact to typed metadata, provenance commitments, privacy limits, and evaluation endpoints. Coalition evaluations then populate a value index with marginal-value estimates, uncertainty, risk flags, and drift evidence. Payments are written to a ledger with formula hashes, budget scaling, and evidence pointers, after which audit logs support disputes and revaluation when the model, retriever, or policy changes. The mechanism-design layer is therefore not isolated game theory; it is the pricing logic inside a deployment-style data system for registering, valuing, paying, and auditing foundation-model adaptation artifacts.

\textbf{Paper organization.} Section~\ref{sec:related} surveys related work. Section~\ref{sec:formulation} formalizes the FedFM market and pipeline-ordered valuation, including its uniqueness characterization. Section~\ref{sec:system} presents the system architecture, Section~\ref{sec:s3val} the \valuator{} estimator, and Section~\ref{sec:payment} the risk-adjusted payment mechanism and its guarantees. Section~\ref{sec:bench} introduces \bench{}, and Section~\ref{sec:eval} reports the real-data evaluation. Section~\ref{sec:discussion} discusses limitations and scope, and Section~\ref{sec:conclusion} concludes.

\section{Related Work}
\label{sec:related}

\subsection{Federated Learning Incentives}

Federated learning enables collaborative model training without centralizing raw data, but its practical adoption depends on incentives for participation, computation, communication, and risk sharing \cite{kairouz2021advancesfl}. Shapley-based reward allocation is widely studied because it satisfies fairness axioms in cooperative games \cite{shapley1953value,yang2023flshapleypareto}. Other work studies adaptive contribution scoring, seller selection, auctions, Stackelberg games, and collaborative fairness in FL markets \cite{zhang2024shapleyucb,wang2024fedave,zhang2025fmincentive}. Classical VCG-style mechanisms and peer-prediction mechanisms offer stronger truthfulness results under restrictive observability and reporting assumptions \cite{vickrey1961counterspeculation,clarke1971multipart,groves1973incentives,miller2005peer}. These mechanisms assume observability and reporting conditions that FedFM artifacts do not meet: artifacts are privacy-constrained, non-verifiable in raw form, and pipeline-dependent. \method{} therefore uses auditable approximate valuation with explicit uncertainty and dispute logs, and treats VCG as a full-observability reference point: VCG attains higher utility when artifacts and utilities are fully observable, whereas \method{} is built for the privacy-constrained case where raw observability is unavailable. Closest in spirit are privacy-aware FL auctions such as \emph{FL-Market}, which trade models under local differential privacy with optimal aggregation \cite{jiang2021flmarket}, and sybil-poisoning defenses such as \emph{FoolsGold}, which down-weight low-diversity gradient updates \cite{fung2018foolsgold}; both, however, target gradient or dataset providers under a supervised objective, whereas \method{} prices \emph{typed} FM artifacts under pipeline causality and can host a gradient-diversity defense inside its update-sketch track rather than replace it. These approaches provide important foundations, but they typically treat each client as contributing a dataset or update under a relatively fixed supervised-learning objective. FedFMs introduce richer artifact types and pipeline dependencies.

\subsection{Data Valuation}

Data valuation estimates the contribution of examples, sources, or clients to a model or task. Data Shapley values allocate utility by averaging marginal contributions over coalitions \cite{ghorbani2019datashapley}. Efficient variants exploit nearest-neighbor structure, gradient similarity, influence approximations, or surrogate models \cite{jia2019efficient,koh2017influence}. DVRL and noise-reduced Shapley variants provide additional peer-reviewed baselines for non-market valuation \cite{yoon2020dvrl,kwon2022betashapley}. Recent foundation-model valuation work studies document credit in LLM summaries and data auctions for RAG \cite{ye2025documentvaluation,han2025ragauctions}. In foundation-model applications, valuation must also handle retrieval, in-context learning, fine-tuning, preference optimization, safety data, and unlearning audits. \method{} builds on Shapley-style marginal value but treats valuation as an operational layer over heterogeneous artifacts and market logs.

\subsection{Federated Foundation Models}

Federated foundation models combine the general capabilities of foundation models with privacy-preserving multi-client collaboration. Recent surveys and challenge papers identify private data use, non-IID heterogeneity, bidirectional knowledge transfer, incentives, game mechanisms, privacy, security, watermarking, and efficiency as open problems \cite{fan2025fedfmchallenges}. FM-enabled cross-silo incentive work has begun to study knowledge hoarding, free-riding, and Stackelberg-style compensation \cite{zhang2025fmincentive}. This paper targets the data-engineering part of that agenda: how to price heterogeneous artifacts that improve a shared foundation-model pipeline without centralizing raw private data.

\subsection{Foundation-Model Adaptation Artifacts}

Modern foundation-model systems can be adapted through retrieval-augmented generation (RAG), prompt engineering, in-context demonstrations, LoRA or other parameter-efficient adapters, preference data, safety data, and tool traces \cite{lewis2020rag,hu2022lora,ouyang2022training}. Federated PEFT and LoRA studies show that adapter rank, aggregation, initialization, robustness, privacy noise, and heterogeneity matter in collaborative foundation-model tuning \cite{cho2024hetlora,chen2024rolora,singhal2025fedexlora,bian2024lorafair,wang2024flora,huang2026sfedlora,kou2026winflora}. Adapter merging and routing policies can change coalition interference, so our framework includes proxy comparisons for centroid merging, volume routing, risk-aware routing, orthogonal-dropout-style merging, and hierarchical merging. These works primarily improve aggregation or global-model accuracy. \method{} is complementary: a secure evaluator can run any of these aggregation rules, while the market layer values the submitted adapters, charges for privacy/cost/risk, and records payment evidence.

\subsection{Collaborative RAG and Federated Seller Measurement}

Collaborative RAG studies show that shared passage stores can improve low-resource clients, but that irrelevant or hard-negative passages affect both performance and participation incentives \cite{muhamed2025corag}. Privacy-preserving federated RAG such as \emph{FedE4RAG} learns retriever embeddings under privacy constraints \cite{fede4rag2025}; \method{} is compatible with such systems, treating a federated retriever as an upstream artifact to be valued and paid. Decentralized data-market work on federated data measurements ranks sellers with relevance and diversity signals without training task-specific models \cite{lu2024datameasurements}. \method{} treats these measurements as useful priors rather than replacements for valuation: relevance/diversity sketches initialize redundancy clusters and sampling priorities, while secure coalition evaluation still determines auditable payments under duplicate, poison, and scarcity adjustments.

\subsection{Structure-Aware Valuation and Secure Incentive Systems}

Recent structure-aware valuation work argues that classical Shapley symmetry can be inappropriate when data sources enter ordered pipelines. Asymmetric Data Shapley relaxes symmetry to respect group or temporal precedence \cite{zheng2024ads,zheng2025ads}, and precedence-constrained Winter values extend constrained-permutation valuation to graph dependencies \cite{chi2024pcwinter}; our pipeline-ordered value is a serving-pipeline instance of this broader family of constrained-permutation valuations. This is directly relevant to foundation-model systems, where retrieval corpora, prompts, adapters, and preference/safety data may be evaluated in an ordered serving path. Efficient valuation work such as Owen-sampling-based FL contribution estimation and data-free spectral-entropy metrics provide low-cost alternatives or priors for S3Val's sampling and sketch layers \cite{li2025fedowen,entropy2026datafree}. Secure system work such as C-FedRAG focuses on confidential federated retrieval, while FWeb3 focuses on incentive-aware FL settlement with Web3 support services \cite{addison2024cfedrag,yan2026fweb3}. \method{} is complementary: it supplies typed artifact valuation, risk-adjusted payments, DP/enclave evidence, and dispute ledgers that can sit above confidential retrieval or settlement substrates.

Table~S8 summarizes the gap. The closest prior lines each solve part of the problem, but \method{} treats the market itself as a deployable data-management layer for foundation-model artifacts.

\section{Problem Formulation}
\label{sec:formulation}

\subsection{Market Participants}

The market contains a set of clients \(N=\{1,\ldots,n\}\), a market operator, downstream consumers, and optionally an external auditor. Each client \(i\) owns a private portfolio \(\Pi_i=\{z_{i1},\ldots,z_{ik_i}\}\). A contribution \(z\) is a typed artifact:

\[
z=(type, payload, metadata, policy).
\]

The payload may remain private. The operator may observe signed hashes, provenance commitments, differentially private summaries, adapter fingerprints, secure evaluation outputs, or local evaluation certificates. The operator should not need raw private corpora, private test sets, or sensitive user records.

\subsection{Contribution Types}

Table~\ref{tab:contributions} summarizes the main contribution classes, and Figure~\ref{fig:taxonomy} shows the corresponding heterogeneous contribution taxonomy. A full deployment binds each artifact type to a separate secure evaluation interface.

\begin{table}[t]
\centering
\caption{FedFM market contribution types.}
\label{tab:contributions}
\footnotesize
\begin{tabular}{@{}p{0.22\linewidth}p{0.25\linewidth}p{0.33\linewidth}@{}}
\toprule
\rowcolor{headrow} Type & Example & Evaluation interface \\
\midrule
Retrieval corpus & Private passages & Federated retriever \\
Adapter & LoRA checkpoint & Adapter routing/merge \\
Prompt & Template & Contract prompt probes \\
Demonstration & ICL examples & Context selection \\
Preference data & Pairwise labels & Reward/alignment probes \\
Safety data & Red-team cases & Hidden safety tests \\
Update sketch & FL gradient/update & Secure aggregation \\
\bottomrule
\end{tabular}
\end{table}

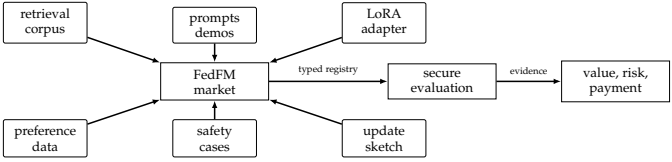
\begin{figure}[t]
\centering
\resizebox{\linewidth}{!}{%
\begin{tikzpicture}[
  artifact/.style={draw, rounded corners=1pt, align=center, font=\scriptsize, minimum width=1.55cm, minimum height=0.55cm},
  layer/.style={draw, align=center, font=\scriptsize, minimum width=2.0cm, minimum height=0.65cm},
  arrow/.style={-{Latex[length=1.4mm]}, thick}
]
\node[layer] (market) {FedFM\\market};
\node[artifact, above left=0.35cm and 1.35cm of market] (rag) {retrieval\\corpus};
\node[artifact, above=0.35cm of market] (prompt) {prompts\\demos};
\node[artifact, above right=0.35cm and 1.35cm of market] (adapter) {LoRA\\adapter};
\node[artifact, below left=0.35cm and 1.35cm of market] (pref) {preference\\data};
\node[artifact, below=0.35cm of market] (safe) {safety\\cases};
\node[artifact, below right=0.35cm and 1.35cm of market] (update) {update\\sketch};
\node[layer, right=2.2cm of market] (eval) {secure\\evaluation};
\node[layer, right=1.2cm of eval] (pay) {value, risk,\\payment};
\draw[arrow] (rag)--(market);
\draw[arrow] (prompt)--(market);
\draw[arrow] (adapter)--(market);
\draw[arrow] (pref)--(market);
\draw[arrow] (safe)--(market);
\draw[arrow] (update)--(market);
\draw[arrow] (market)--node[above,font=\tiny]{typed registry} (eval);
\draw[arrow] (eval)--node[above,font=\tiny]{evidence} (pay);
\end{tikzpicture}
}
\caption{Heterogeneous contribution taxonomy: RAG corpora, prompts, demonstrations, adapters, preference data, safety cases, and update sketches as typed market artifacts.}
\label{fig:taxonomy}
\end{figure}

\subsection{Contract Utility}

For a coalition \(C \subseteq N\), let \(\Pi(C)=\cup_{i\in C}\Pi_i\). The market utility is:

\begin{align}
U(C) ={}& w_q Q(C) + w_s S(C) + w_r R(C) \nonumber \\
&- w_l L(C) - w_p \Phi(C) - w_k K(C),
\label{eq:utility}
\end{align}

where \(Q\) is task quality, \(S\) is safety compliance, \(R\) is robustness, \(L\) is latency or serving overhead, \(\Phi\) is privacy-budget consumption or leakage risk, and \(K\) is financial/compute cost. The weights are contract parameters and must be reported as part of the benchmark card.

\subsection{Contribution Value}

The ideal client value is the Shapley-style marginal contribution:

\begin{equation}
\phi_i = \mathbb{E}_{\omega}\left[ U(Pre_{\omega}(i)\cup\{i\}) - U(Pre_{\omega}(i)) \right],
\label{eq:clientvalue}
\end{equation}

where \(\omega\) is a random client ordering and \(Pre_{\omega}(i)\) are clients before \(i\). Artifact-level values are defined analogously for each \(z_{ij}\). In practice, direct computation is infeasible because it requires many RAG indexes, adapter compositions, prompt evaluations, or secure aggregation rounds.

\subsection{Pipeline-Ordered Contribution Value}

The exchangeability implicit in Eq.~\eqref{eq:clientvalue} is often wrong for foundation-model serving. A RAG passage is upstream of a prompt, an adapter acts after retrieved context is selected, and preference or safety data may only matter after the answer policy is invoked. We therefore define a pipeline-ordered value, \(\phi_i^{ord}\), by restricting permutations to respect a partial order over artifact groups:
\[
\begin{aligned}
\mathcal{G}: \quad &\mathrm{retrieval}\prec \mathrm{prompt/demo}\\
&\prec \mathrm{adapter}\prec \mathrm{preference/safety}.
\end{aligned}
\]
Let \(g(i)\in\mathcal{G}\) be the pipeline group of client \(i\). Let \(\Omega_{\mathcal{G}}\) be the set of client permutations in which all clients from earlier groups appear before clients from later groups; clients inside the same group remain randomly ordered. The ordered value is
\begin{equation}
\phi_i^{ord} =
\mathbb{E}_{\omega\sim\Omega_{\mathcal{G}}}
\left[
U(Pre_{\omega}^{\mathcal{G}}(i)\cup\{i\})
-U(Pre_{\omega}^{\mathcal{G}}(i))
\right].
\label{eq:orderedvalue}
\end{equation}
This formulation is an Asymmetric-Data-Shapley-style relaxation of symmetry for FedFM markets: it preserves marginal-credit accounting but uses the serving path as a contract-visible precedence constraint. The practical implication is simple. A retrieval client is credited for improving the context available to downstream prompts and adapters, while an adapter is credited conditional on the retrieval/prompt state that would actually precede it at serving time. The contract card records whether a market round uses unordered S3Val or pipeline-ordered S3Val, and the ledger stores the group order so disputed payments can be replayed.

\subsection{Axiomatic Characterization of Ordered Valuation}
\label{sec:ordered-axioms}

Equation~\eqref{eq:orderedvalue} is one way to respect serving order; we now show it is the \emph{principled} one. We characterize \(\phi^{ord}\) as the unique credit rule satisfying four standard value axioms plus one market axiom that encodes serving causality. This turns pipeline-ordered valuation from a heuristic relaxation into the Shapley value for the layered precedence structure induced by \(\mathcal{G}\), in the sense of games under precedence (permission) constraints \cite{faigle1992precedence,owen1977values}.

A \emph{pipeline game} is a tuple \((N,U,\mathcal{G})\) with \(U(\emptyset)=0\) and \(\mathcal{G}=(G_1\prec\cdots\prec G_L)\) an ordered partition of \(N\) into serving layers, where \(g(i)\) is the layer index of client \(i\). A \emph{value} \(\psi\) maps each pipeline game to a payoff vector \((\psi_i)_{i\in N}\). For \(\emptyset\neq T\subseteq N\), the unanimity (pure-complementarity) game is \(u_T(S)=\mathbf{1}[T\subseteq S]\), and \(\ell(T)=\max\{g(k):k\in T\}\) is the most downstream layer that \(T\) reaches. We impose, on the class of pipeline games with fixed \(\mathcal{G}\):

\begin{itemize}
\item[(A1)] \textbf{Linearity.} \(\psi(\alpha U+\beta V)=\alpha\psi(U)+\beta\psi(V)\).
\item[(A2)] \textbf{Null player.} If \(U(S\cup\{i\})=U(S)\) for all \(S\), then \(\psi_i(U)=0\).
\item[(A3)] \textbf{Efficiency.} \(\sum_{i\in N}\psi_i(U)=U(N)\).
\item[(A4)] \textbf{Within-layer symmetry.} If \(i,j\in G_a\) and \(U(S\cup\{i\})=U(S\cup\{j\})\) for all \(S\subseteq N\setminus\{i,j\}\), then \(\psi_i(U)=\psi_j(U)\).
\item[(A5)] \textbf{Downstream realization.} For every \(T\), any member \(i\in T\) with \(g(i)<\ell(T)\) receives none of \(T\)'s pure complementarity: \(\psi_i(u_T)=0\).
\end{itemize}

Axiom~(A5) is the serving-causal content: when a coalition's value is realized only once all members are assembled and served, the marginal is emitted at the serving frontier \(G_{\ell(T)}\), so an upstream member that shares that synergy with a strictly-downstream partner draws no credit from it. The symmetric Shapley value violates~(A5); it splits every complementarity dividend equally across all of \(T\). We state~(A5) as a deliberate contract-level \emph{design} choice about serving-time attribution---the market elects to credit synergy where it is realized---rather than a claim that upstream inputs are normatively worthless; an operator who prefers to reward necessary upstream inputs can simply run unordered \valuator{}, and the contract card records which rule is in force.

\begin{theorem}[Serving-order characterization]
\label{thm:ordered-char}
On pipeline games with layered precedence \(\mathcal{G}\), \(\phi^{ord}\) is the unique value satisfying (A1)--(A5). Moreover it routes each coalition's Harsanyi dividend to that coalition's most downstream members, split equally:
\[
\phi_i^{ord}(u_T)=
\begin{cases}
1/\,|T\cap G_{\ell(T)}|, & i\in T\cap G_{\ell(T)},\\
0, & \text{otherwise.}
\end{cases}
\]
\end{theorem}

\begin{IEEEproof}
\emph{\(\phi^{ord}\) satisfies the axioms.} Marginal contributions are linear in \(U\) and averaging over \(\Omega_{\mathcal{G}}\) preserves linearity, giving (A1). A null player has zero marginal in every order, giving (A2). Each order telescopes, \(\sum_i[U(Pre^{\mathcal{G}}_\omega(i)\cup\{i\})-U(Pre^{\mathcal{G}}_\omega(i))]=U(N)-U(\emptyset)=U(N)\), and averaging preserves the sum, giving (A3). Transposing two clients in the same layer is a measure-preserving bijection of \(\Omega_{\mathcal{G}}\), so symmetric same-layer clients receive equal value, giving (A4). Finally, in \(u_T\) a client \(i\) has marginal \(1\) in order \(\omega\) iff \(i\in T\) and every other member of \(T\) precedes \(i\), i.e.\ \(i\) is the last member of \(T\) in \(\omega\); under \(\Omega_{\mathcal{G}}\) that last member lies in \(G_{\ell(T)}\) and is uniform over \(T\cap G_{\ell(T)}\), which yields the displayed formula and in particular (A5).

\emph{Uniqueness.} The unanimity games \(\{u_T\}_{\emptyset\neq T\subseteq N}\) form a basis of the space of games with \(U(\emptyset)=0\); write \(U=\sum_T c_T u_T\). By (A1), \(\psi_i(U)=\sum_T c_T\psi_i(u_T)\), so \(\psi\) is fixed once \(\{\psi(u_T)\}\) is fixed. Take any \(T\). Clients \(i\notin T\) are null in \(u_T\), so (A2) gives \(\psi_i(u_T)=0\); clients \(i\in T\) with \(g(i)<\ell(T)\) give \(\psi_i(u_T)=0\) by (A5). The remaining clients lie in \(T\cap G_{\ell(T)}\), share the layer \(G_{\ell(T)}\), and are mutually symmetric in \(u_T\), so (A4) equates their values to a common \(s\). Efficiency (A3) forces \(|T\cap G_{\ell(T)}|\,s=u_T(N)=1\), hence \(s=1/|T\cap G_{\ell(T)}|\). Thus \(\psi(u_T)=\phi^{ord}(u_T)\) for every \(T\), and by linearity \(\psi=\phi^{ord}\).
\end{IEEEproof}

Replacing (A4)--(A5) by full symmetry recovers the classical Shapley value, which splits each \(c_T\) equally over all of \(T\); ordered valuation instead routes \(c_T\) to the serving frontier. Two consequences matter for the market. First, because Theorem~\ref{thm:ordered-char} fixes the value on the unanimity basis, contracts that agree on \(\mathcal{G}\) yield identical ordered credit, and the budget-feasible scaling of the payment rule in Eq.~\eqref{eq:payment} multiplies all positive values by a common factor, preserving the within-layer order the theorem induces. Second, combined with the held-out serving-neutrality result (Table~S15), ordered valuation rests on two pillars: it is the \emph{unique} serving-causal credit rule, and it costs no measurable held-out task quality. It is therefore a principled payment-fairness guarantee rather than an accuracy heuristic.

\noindent\textbf{Worked example.} Consider two retrieval clients \(r_1,r_2\) in layer \(G_1\) and one adapter client \(a\) in layer \(G_2\), with \(U(\emptyset)=0\), \(U(\{r_1\})=U(\{r_2\})=2\), \(U(\{a\})=0\), \(U(\{r_1,r_2\})=3\), \(U(\{r_1,a\})=U(\{r_2,a\})=5\), and \(U(N)=6\): the adapter is useless without retrieved context but adds \(3\) once either retrieval client is present. The nonzero Harsanyi dividends are \(c_{\{r_i\}}=2\), \(c_{\{r_1,r_2\}}=-1\), \(c_{\{r_i,a\}}=3\), and \(c_N=-3\). The symmetric Shapley value splits each dividend across all of its members and pays \((2,2,2)\). Ordered valuation instead routes the two cross-layer complementarity dividends \(c_{\{r_i,a\}}\) entirely to the serving frontier \(a\), giving \((1.5,1.5,3)\) (Table~\ref{tab:worked-example}): the adapter that realizes the retrieval--adapter synergy at serving time is paid for it, while the retrieval clients keep only their standalone and same-layer credit. Both rules are efficient (\(\sum_i\phi_i=6\)); they differ precisely on who is paid for cross-layer complementarity. Averaging the two serving-order permutations \((r_1,r_2,a)\) and \((r_2,r_1,a)\) in Eq.~\eqref{eq:orderedvalue} reproduces \((1.5,1.5,3)\), confirming the dividend-routing form of Theorem~\ref{thm:ordered-char}.

\begin{table}[t]
\centering
\caption{Worked example: symmetric versus ordered credit for a retrieval\(\prec\)adapter game.}
\label{tab:worked-example}
\scriptsize
\begin{tabular}{lrr}
\toprule
\rowcolor{headrow} Client (layer) & Symmetric Shapley & Ordered value \\
\midrule
\(r_1\) (retrieval) & 2.0 & 1.5 \\
\(r_2\) (retrieval) & 2.0 & 1.5 \\
\(a\) (adapter) & 2.0 & \textbf{3.0} \\
\midrule
\rowcolor{keyrow} Total & 6.0 & 6.0 \\
\bottomrule
\end{tabular}
\end{table}

Beyond attribution, (A5) has a strategic consequence that distinguishes ordered from symmetric credit.

\begin{corollary}[No upstream synergy capture]
\label{cor:no-upstream}
Under \(\phi^{ord}\), client \(i\)'s value draws only on coalitions whose deepest layer is \(i\)'s own:
\(\phi_i^{ord}(U)=\sum_{T:\,i\in T,\ \ell(T)=g(i)} c_T/|T\cap G_{g(i)}|\).
Hence \(i\) receives nothing from any coalition that contains a strictly more downstream member; a positive cross-layer complementarity dividend \(c_T\) with \(\ell(T)>g(i)\) pays \(i\) a share \(c_T/|T|\) under the symmetric Shapley value but \(0\) under \(\phi^{ord}\).
\end{corollary}

\begin{IEEEproof}
By linearity, \(\phi_i^{ord}(U)=\sum_T c_T\,\phi_i^{ord}(u_T)\); Theorem~\ref{thm:ordered-char} makes \(\phi_i^{ord}(u_T)\) nonzero only when \(g(i)=\ell(T)\), which gives the displayed sum and the vanishing of every term with \(\ell(T)>g(i)\).
\end{IEEEproof}

Two operational consequences follow. First, an upstream provider cannot raise its ordered payment by entering or fabricating cross-layer complementarity, so identity-splitting or collusion aimed at harvesting downstream synergy is unprofitable under ordered credit even before the duplicate penalty of Section~\ref{sec:payment} is applied. Second, credit---and therefore the manipulation surface---concentrates on the serving frontier, which tells the operator to prioritize duplicate and poison audits on the most downstream layer a disputed coalition reaches. A full incentive-compatibility analysis of frontier-concentrated credit is left to future work.

\subsection{Strategic Clients}

Clients may split one portfolio across identities, duplicate another client's artifacts, submit poisoned adapters, overfit public prompts, inflate costs, exaggerate privacy constraints, or collude to create artificial complementarity. \method{} targets approximate manipulation resistance under a bounded strategic model: common manipulations should not produce higher expected profit than honest participation after duplicate penalties, uncertainty discounts, hidden audits, and risk-adjusted payments. This is the appropriate guarantee for privacy-constrained artifacts, whose raw form cannot be verified to support exact dominant-strategy truthfulness.

\section{The \method{} System}
\label{sec:system}

Figure~\ref{fig:architecture} shows the deployed-style architecture, and Figure~\ref{fig:rounddataflow} shows the dataflow within a single market round. The design keeps registry entries, value estimates, payment decisions, and cryptographic evidence as explicit data-engineering artifacts across client, operator, and auditor trust boundaries. The contract card is committed before submissions, and later disputes replay the recorded hashes, utility weights, evaluator evidence, and payment formula.

\begin{figure*}
  \centering
  \includegraphics[width=0.7\linewidth]{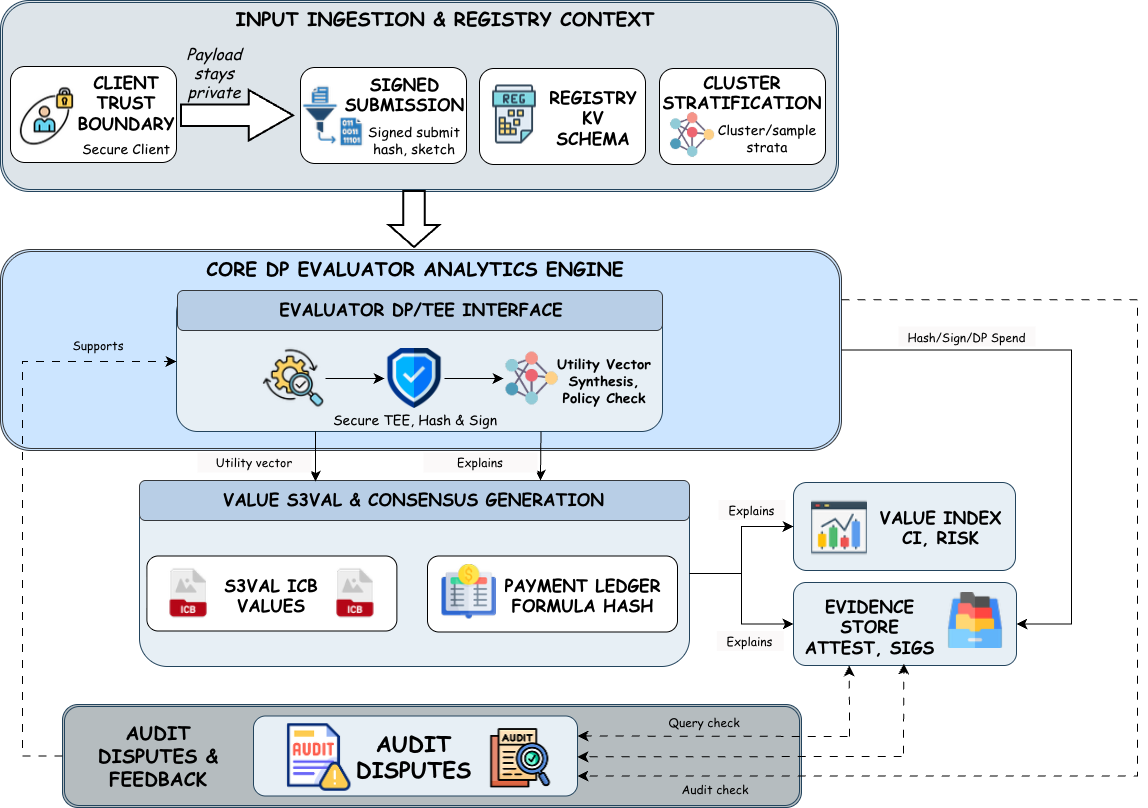}
  \caption{Architecture of \method{}. Solid arrows are data/control flow; dashed arrows are audit/dispute flow.}
  \label{fig:architecture}
\end{figure*}

\begin{figure}
  \centering
  \includegraphics[width=1\linewidth]{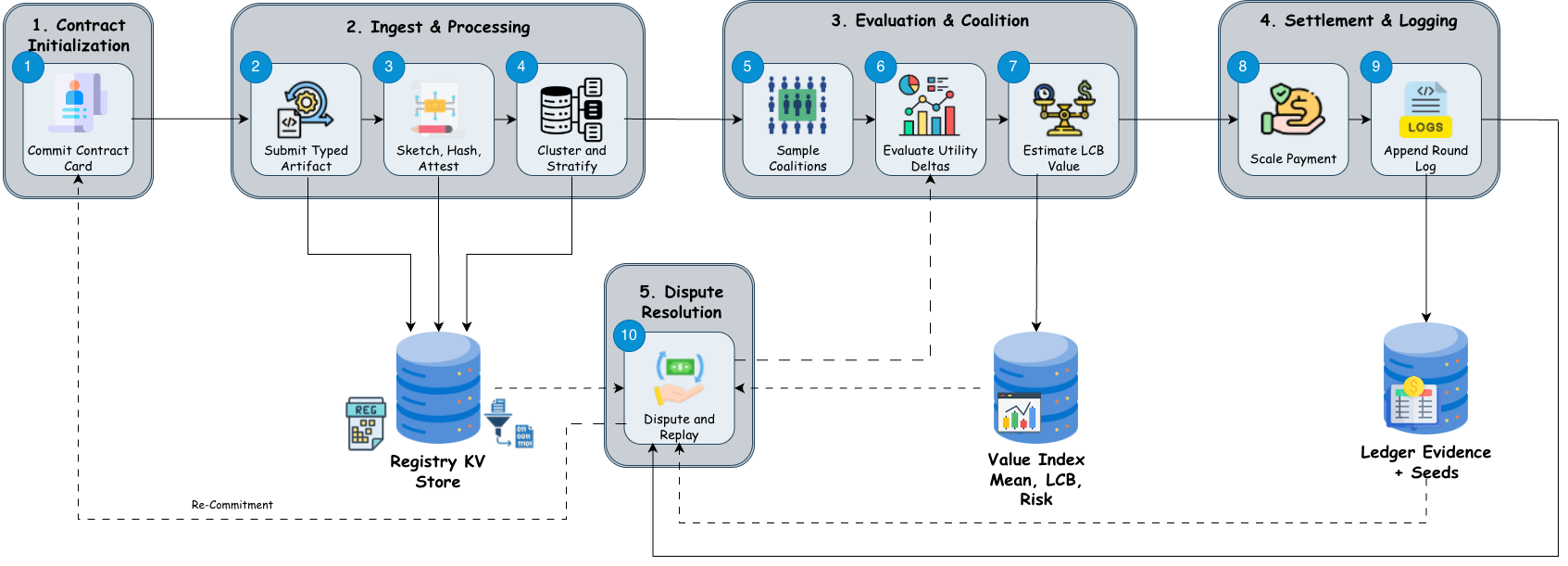}
  \caption{Single market-round dataflow, from contract-card commitment through settlement and dispute replay.}
  \label{fig:rounddataflow}
\end{figure}

\subsection{Contribution Registry}

The registry stores client credentials, artifact type, domain tags, signed hashes, license constraints, privacy budget, declared cost, and artifact endpoint. The registry is not only a database; it defines what can be evaluated and what evidence can be used in disputes.

\subsection{Secure Evaluation Sandbox}

The sandbox evaluates sampled coalitions. For RAG, it builds temporary indexes or queries federated retriever endpoints. For adapters, it loads or routes LoRA modules in isolated workers. For private probes, it accepts signed client-side score certificates. For update sketches, it uses secure aggregation outputs. Each evaluation produces an evidence hash and utility vector.

\subsection{Value Index and Ledger}

The value index stores point estimates, confidence intervals, redundancy clusters, complementarity links, risk flags, and historical value drift. The market ledger stores the payment formula version, utility weights, lower confidence bounds, penalties, bonuses, final payment, and evidence hashes. This is essential for dispute resolution and revaluation after model, prompt, retriever, or policy changes.

\section{\valuator{}: Secure Surrogate Shapley Valuation}
\label{sec:s3val}

\subsection{Overview}

\valuator{} estimates either the unordered value in Eq.~\eqref{eq:clientvalue} or the pipeline-ordered value in Eq.~\eqref{eq:orderedvalue}. The same evidence path is used in both cases; the ordered variant changes only the admissible coalition permutations. The estimator has four components:

\begin{enumerate}
  \item contribution sketches;
  \item redundancy and complementarity clustering;
  \item stratified or pipeline-ordered coalition sampling;
  \item surrogate utility modeling with uncertainty-triggered direct audits.
\end{enumerate}

\subsection{Contribution Sketches}

Each client submits typed metadata and privacy-preserving summaries. A retrieval corpus may submit centroid embeddings, MinHash sketches, domain histograms, and provenance hashes. An adapter may submit rank, target modules, delta norms, calibration traces, and routing hints. A prompt contribution may submit task tags and response statistics. These sketches provide features for valuation while limiting raw data exposure.

\subsection{Coalition Sampling}

Coalitions are sampled by artifact type, domain, redundancy cluster, and risk stratum. The sampler oversamples clients whose value confidence interval crosses zero, whose payment would be large, or whose manipulation risk is high. This matters because a market does not need uniformly precise values for every client; it needs reliable values near payment decisions.

\subsection{Utility Surrogate}

The surrogate \(g_{\theta}(C,i)\) predicts:

\begin{equation}
\Delta_i(C) = U(C\cup\{i\}) - U(C).
\end{equation}

Features include coalition size, artifact mix, domain gaps, duplicate overlap, adapter interference, safety risk, privacy budget, cost, and scarcity. The prototype uses a dependency-free linear surrogate to keep experiments reproducible; the framework can replace it with calibrated gradient-boosted trees, Gaussian processes, or neural models. Algorithm~\ref{alg:s3val} summarizes the full \valuator{} procedure.

\begin{algorithm}[t]
\DontPrintSemicolon
\caption{\valuator{}}
\label{alg:s3val}
\KwIn{Clients \(N\), utility \(U\), strata \(\mathcal{S}\), call budget \(M\), width target \(\tau\), stability window \(T\), audit predicate \(A(i,C,H)\)}
\KwOut{Value ledger}
Collect contribution sketches and provenance commitments\;
Build redundancy and domain clusters\;
Initialize direct set \(D\leftarrow\emptyset\), surrogate \(g_\theta\), widths \(w_i\leftarrow\infty\)\;
\While{\(|D|<M\) and \(\max_i w_i>\tau\) and top-\(k\) payments are not stable for \(T\) updates}{
  Choose stratum \(s\in\mathcal{S}\) with probability proportional to \(1+\bar{w}_s+\bar{r}_s+\bar{b}_s\)\;
  Sample client \(i\) from \(s\) proportional to \(1+w_i+r_i+b_i\), where \(b_i\) is payment-boundary proximity\;
  Sample coalition \(C\subseteq N\setminus\{i\}\) from same/different domain and redundancy strata\;
  \If{\(A(i,C,H)=1\), where \(H\) is the current ledger history}{
    Evaluate \(\Delta_i(C)=U(C\cup\{i\})-U(C)\); append \((x(C,i),\Delta_i(C))\) to \(D\)\;
  }
  Update \(g_\theta\) by ridge regression on \(D\); update residual and marginal-variance widths \(w_i\)\;
}
\ForEach{client \(i\)}{
  Estimate \(\widehat{\phi}_i\) from direct and surrogate marginals\;
  Compute uncertainty, lower confidence bound, duplicate risk, and manipulation risk\;
}
\Return{value ledger}
\end{algorithm}

With fixed strata and bounded marginal variance, a client receiving \(m_i\) direct marginal observations has standard error \(O(1/\sqrt{m_i})\). Thus the LCB width shrinks at the same rate until the stopping rule is met or the call budget is exhausted. The audit predicate is explicit:
\[
A(i,C,H)=\mathbf{1}\{w_i>\tau_w \vee r_i>\tau_r \vee b_i>\tau_b \vee \xi_i<h_{\min}\},
\]
where \(b_i\) is the distance of client \(i\)'s current scaled payment from the budget threshold, \(\xi_i\sim \mathrm{Uniform}(0,1)\), and \(h_{\min}\) is the configured random direct-audit floor. The direct-evaluation cost is \(2|D|\leq 2M\) coalition utility calls, plus sketch clustering and surrogate updates; for \(k\) strata, a balanced allocation gives \(O(M/k)\) direct samples per stratum.

\section{Payment Mechanism}
\label{sec:payment}

\subsection{Lower-Bound Payment Rule}

Let \(\widehat{\phi}_i\) be the estimated value, \(\sigma_i\) its calibrated uncertainty, \(d_i\) duplicate risk, \(r_i\) manipulation risk, \(c_i\) verified or declared cost, \(\epsilon_i^{priv}\) privacy budget, and \(s_i\) scarcity score. The raw payment is the positive part of a single contract score:
\begin{equation}
\widetilde{p}_i =
\left[
\widehat{\phi}_i-\lambda\sigma_i
-\beta c_i-\gamma \epsilon_i^{priv}
-\eta\max(d_i,r_i)+\rho s_i
\right]_+ .
\label{eq:payment}
\end{equation}
Here \(\lambda\) is the uncertainty discount, \(\beta\) the cost penalty, \(\gamma\) the privacy-consumption penalty, \(\eta\) the risk penalty, and \(\rho\) the scarcity bonus weight.

Raw payments \(\widetilde{p}_i\) are scaled if their total exceeds the market budget. Lower-confidence-bound payment is intentional: a high-variance client should not receive a large payment before audit evidence supports the value. Default values are \(\lambda=0.75\), \(\beta=0.28\), \(\gamma=0.20\), \(\eta=0.75\), and \(\rho=0.25\); Table~S11 lists these defaults, and real-track scripts sweep risk, duplicate, privacy, cost, and scarcity settings and record the selected operating point in the contract card.

\textbf{Operational scores.} In the prototype, \(d_i\) is computed from provenance hash matches when available and otherwise from token/trigram overlap for retrieval passages or train-text signatures for adapter proxies. The manipulation score \(r_i\) is an audit prior built from hidden-probe failure, duplicate-cluster membership, provenance/timing anomalies, declared-cost outliers, and privacy declarations that reduce observability without improving validation utility. The scarcity score \(s_i\) is not a flat entitlement: it is positive only when a domain or artifact type is underrepresented on the contract validation card and the client improves that slice. The privacy term \(\epsilon_i^{priv}\) is the measured DP spend or declared privacy-constrained evaluation budget. The coefficient \(\gamma\) is only the unit conversion from that spend to a payment penalty; it is not itself a privacy budget. The prototype sweeps \(\gamma\) for Laplace and zCDP-style Gaussian releases to test payment stability under different operator cost assumptions.

Cost \(c_i\) should be verified when it affects payment. The deployment interface accepts signed evaluator receipts: wall-clock time, accelerator type, token count, retrieval calls, adapter-train steps, DP releases, and enclave or worker measurement. If receipts are unavailable, the prototype treats cost as adversarially declared and applies reserve checks plus a cost-gaming stress test. In that test, a client jointly inflates cost by \(1\times\)--\(10\times\) and submits a duplicate; the attacker is not selected under the current contract because higher declared cost lowers value density and duplicate/risk penalties dominate.

\subsection{Manipulation Defenses}

Sybil splitting and duplicate submission are handled by redundancy clusters and group-level marginal discounts. Poisoned adapters are penalized through hidden safety probes and manipulation-risk scores. Privacy-budget gaming increases uncertainty and therefore lowers the lower confidence bound. Cost inflation is limited by cost penalties and reserve checks. Non-IID specialist suppression is addressed through stratified evaluation and scarcity bonuses.

\subsection{Theoretical Properties}
We state the mechanism's guarantees informally here and defer formal statements and proofs to the supplementary material. \valuator{} admits a Hoeffding-style estimation-error bound; under a value-gap margin condition, budgeted selection recovers the downstream-oracle set with high probability; budget feasibility and payment monotonicity hold by construction; and approximate incentive-compatibility, sybil, poisoning, and individual-rationality bounds quantify when common manipulations are unprofitable.

\section{\bench{}}
\label{sec:bench}

\subsection{Tracks}

\bench{} contains four tracks:

\begin{itemize}
  \item Federated retrieval corpus market: clients contribute private passages for RAG.
  \item Federated adapter market: clients contribute LoRA or adapter-effect artifacts.
  \item Prompt and demonstration market: clients contribute prompt templates or in-context examples.
  \item Preference and safety market: clients contribute pairwise labels and red-team cases.
\end{itemize}

\subsection{Client Types}

The benchmark includes high-quality specialists, redundant generalists, noisy contributors, sybil splitters, duplicate submitters, poisoned adapter clients, privacy-sensitive clients, cost inflaters, and rare-domain contributors.

\subsection{Metrics}

Valuation metrics include rank correlation with expensive leave-one-client-out, top-\(k\) high-value precision, harmful-client AUROC, duplicate discount ratio, value sign stability, and confidence interval coverage. Market metrics include utility per dollar, budget violation rate, individual rationality rate, diversity of selected domains, and regret against oracle selection. Strategic metrics include sybil gain ratio, duplicate overpayment, poisoning profit, prompt-gaming gain, privacy-gaming gain, and cost-inflation gain.

\subsection{Submission and Leaderboard Protocol}

To make \bench{} more than a one-off reproduction script, a benchmark submission consists of a signed client manifest, typed artifact descriptors, optional sketches, and a score file. Each manifest records \texttt{round\_id}, \texttt{client\_id}, \texttt{artifact\_type}, \texttt{base\_model}, \texttt{provenance\_hash}, \texttt{privacy\_budget}, \texttt{declared\_cost}, \texttt{endpoint\_ref}, and \texttt{signature}. Raw private payloads are not submitted to the leaderboard. Held-out evaluation uses a contract card whose validation-slice hashes and audit policy are committed before submissions; test labels, hidden probes, and rare-slice membership are not exposed until the round closes. A leaderboard row reports per-track raw utility, risk-adjusted utility, strategic clients selected, rare clients retained, audit calls, model calls, wall-clock time, budget use, and dispute count. \bench{} includes JSON schemas and toy clients for specialist, duplicate, sybil, poison, privacy-gaming, and cost-inflation submissions so outside researchers can test new market mechanisms against the same interface. Figure~\ref{fig:benchflow} shows the resulting submission flow.

\textbf{A runnable, extensible leaderboard.} To lower the barrier to entry, \bench{} ships as a small library and a one-command harness. An external method is a \emph{scorer} that returns a contribution score per client, given a committed contract card and a validation-utility oracle that exposes the same secure-evaluation interface every baseline uses. The harness selects a budget-feasible coalition from the returned scores, serves it on a disjoint held-out test card, and appends a standardized row; a JSON submission schema and the frozen contract-card hash make results comparable and commit-reveal-checkable. Table~\ref{tab:leaderboard} shows the leaderboard pre-populated with the seven baselines on a frozen 100-client instance (contract \texttt{72a9f6c9}). Beyond task quality it reports operator economics at a fixed budget: \method{} attains the best held-out accuracy and utility-per-dollar while retaining the rare specialist and selecting no strategic clients, whereas volume and leave-one-out spend the same budget on high-volume poison (10--12 strategic clients) and collapse. Because selection and serving are decoupled by the contract card, a new valuation method is scored against the identical instance in a single call.

\begin{table}[t]
\centering
\caption{\bench{} leaderboard on a frozen 100-client instance: held-out accuracy and operator economics at a fixed budget.}
\label{tab:leaderboard}
\scriptsize
\setlength{\tabcolsep}{3pt}
\begin{tabular}{lrrrrr}
\toprule
\rowcolor{headrow} Method & Acc. & Util/\$ & Rare kept & Strat.\ sel. & Calls \\
\midrule
\rowcolor{keyrow} \method{} & \textbf{0.550} & \textbf{0.126} & 1.00 & 0.00 & 351 \\
Equal & 0.517 & 0.116 & 1.00 & 0.00 & 0 \\
Shapley-UCB~\cite{zhang2024shapleyucb} & 0.517 & 0.117 & 1.00 & 0.33 & 351 \\
Retrieval similarity~\cite{lu2024datameasurements} & 0.442 & 0.100 & 0.67 & 0.33 & 0 \\
Leave-one-out~\cite{koh2017influence} & 0.108 & 0.025 & 0.33 & 10.67 & 117 \\
Volume & 0.092 & 0.020 & 0.00 & 12.00 & 0 \\
FL-Shapley~\cite{ghorbani2019datashapley} & 0.075 & 0.017 & 0.00 & 4.67 & 351 \\
\bottomrule
\end{tabular}
\end{table}

% \begin{figure}[t]
% \centering
% \resizebox{\linewidth}{!}{%
% \begin{tikzpicture}[
%   benchstep/.style={draw, rounded corners=1pt, align=center, font=\scriptsize, minimum width=1.65cm, minimum height=0.58cm},
%   store/.style={draw, align=center, font=\scriptsize, minimum width=1.65cm, minimum height=0.58cm},
%   arrow/.style={-{Latex[length=1.4mm]}, thick},
%   audit/.style={-{Latex[length=1.4mm]}, dashed, thick}
% ]
% \node[benchstep] (card) {contract\\hash};
% \node[benchstep, right=0.55cm of card] (manifest) {client\\manifest};
% \node[store, right=0.55cm of manifest] (eval) {held-out\\evaluator};
% \node[store, right=0.55cm of eval] (board) {leaderboard\\record};
% \node[benchstep, below=0.45cm of eval] (dispute) {dispute\\replay};
% \draw[arrow] (card)--(manifest);
% \draw[arrow] (manifest)--node[above,font=\tiny]{schema check} (eval);
% \draw[arrow] (eval)--node[above,font=\tiny]{metrics} (board);
% \draw[audit] (board)--(dispute);
% \draw[audit] (dispute)--(card);
% \end{tikzpicture}
% }
% \caption{\bench{} submission flow. External methods submit signed manifests against a pre-committed contract hash; held-out evaluation produces comparable leaderboard records and dispute replay follows the same evidence path.}
% \label{fig:benchflow}
% \end{figure}

\begin{figure*}
  \centering
  \includegraphics[width=0.85\linewidth]{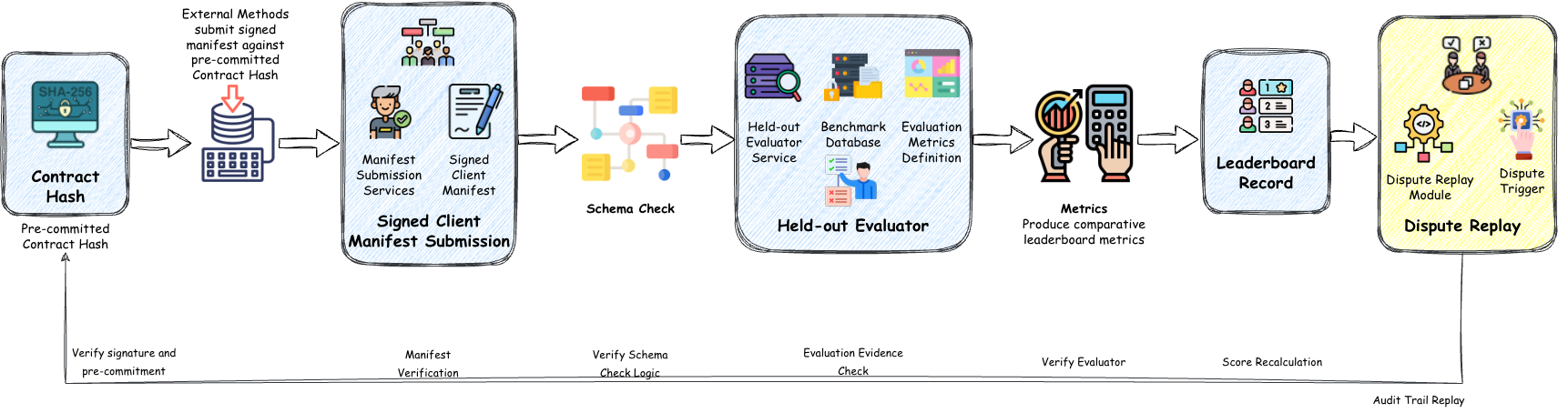}
  \caption{\bench{} submission flow.}
  \label{fig:benchflow}
\end{figure*}

\section{Real-Data Validation}
\label{sec:eval}

\subsection{Tracks and Baselines}

We evaluate on three real-data tracks. The first is a FEVER retrieval-corpus market: clients own private shards of evidence passages and are evaluated by whether a coalition retrieves the gold evidence for held-out claims. The second is a de-circularized FEVER serving track: each method selects a paid coalition using a contract-validation card, then the selected coalition is actually served on a separate held-out test card through retrieve--read scoring. The headline metrics are downstream accuracy, macro-F1, evidence exact match, and regret against a downstream oracle; none of these metrics use attack labels or the payment formula. The third is a trained low-rank adapter validation: each AG News client trains a small PyTorch low-rank adapter, and coalitions are evaluated by summing adapter logits. These tracks exercise real text, real labels, strategic clients, and the major requested baselines.

We also add one optional generator-backed RAG serving track. It uses \texttt{google/flan-t5-small} as a real seq2seq reader after coalition selection: the market selects clients on a validation card, retrieves held-out evidence from the paid coalition, and prompts the generator to answer the FEVER claim as yes/no/unknown, which is mapped back to SUPPORTS/REFUTES/NOT ENOUGH INFO. The attack is deliberately potent: a high-volume poison corpus is useful on the validation card but injects wrong-answer instructions into held-out evidence. This is a realistic prompt-injection threat for RAG systems; \method{} does not use attack-specific rules for it, only the same generic risk, duplicate, uncertainty, and audit scoring used elsewhere. This track is not a large-generator benchmark, but it tests whether the market advantage survives a real generated answer rather than only retrieval scoring.

We compare equal payment, volume, leave-one-out, sampled FL-Shapley, retrieval similarity, Shapley-UCB, and \method{} where applicable. Each rule ranks clients, then selects a budget-feasible subset. Utility includes task performance minus cost, privacy, redundancy, and risk terms. Risk-adjusted utility additionally penalizes selected strategic clients and rewards retained rare specialists, so we report it as an operator composite rather than an independent accuracy metric. Tables therefore keep raw utility, downstream task metrics, strategic selections, rare retention, AUC, and audit/call counts separate wherever possible; the alpha-sweep stress test shows when the composite ordering flips.

\begin{table}[t]
\centering
\caption{Generator-backed held-out FEVER RAG using \texttt{google/flan-t5-small}. Values are mean \(\pm\) 95\% CI.}
\label{tab:generator-backed-rag}
\scriptsize
\resizebox{\linewidth}{!}{%
\begin{tabular}{lcccc}
\toprule
\rowcolor{headrow} Method & Acc. & Macro-F1 & Poison & Strategic \\
\midrule
Volume & \(0.3187{\pm}0.0531\) & \(0.2362{\pm}0.0428\) & 1.0 & 1.0 \\
Leave-one-out & \(0.3125{\pm}0.0490\) & \(0.2269{\pm}0.0411\) & 1.0 & 3.6 \\
FL-Shapley & \(0.3125{\pm}0.0624\) & \(0.2102{\pm}0.0629\) & 0.8 & 2.6 \\
\rowcolor{keyrow} \method{} & \(\mathbf{0.3937{\pm}0.0219}\) & \(\mathbf{0.2718{\pm}0.0198}\) & \textbf{0.0} & \textbf{0.0} \\
\bottomrule
\end{tabular}
}
\end{table}

\begin{table*}[t]
\centering
\caption{Real-data market validation. R-utility is risk-adjusted utility. Spearman is rank correlation with leave-one-out values. AUC measures harmful-client ranking.}
\label{tab:realmarket}
\scriptsize
\begin{tabular}{llccrrr}
\toprule
\rowcolor{headrow} Track & Method & Utility & R-utility & Spearman & AUC & Strategic \\
\midrule
FEVER RAG & Equal & \(0.0219{\pm}0.0308\) & \(-0.1881{\pm}0.0308\) & -0.5874 & 0.5000 & 2.0 \\
FEVER RAG & Volume & \(-0.0386{\pm}0.0308\) & \(-0.2485{\pm}0.0308\) & -0.0489 & 0.7500 & 2.0 \\
FEVER RAG & Leave-one-out & \(0.0096{\pm}0.0308\) & \(-0.2004{\pm}0.0308\) & \textbf{1.0000} & 1.0000 & 2.0 \\
FEVER RAG & FL-Shapley & \(0.0096{\pm}0.0308\) & \(-0.2004{\pm}0.0308\) & 0.8881 & 1.0000 & 2.0 \\
FEVER RAG & Retrieval similarity & \(0.0158{\pm}0.0393\) & \(-0.1942{\pm}0.0393\) & 0.4405 & 0.7031 & 2.0 \\
FEVER RAG & Shapley-UCB & \(0.0219{\pm}0.0308\) & \(-0.1881{\pm}0.0308\) & 0.5944 & 0.6719 & 2.0 \\
\rowcolor{keyrow} FEVER RAG & \method{} & \(\mathbf{0.1321{\pm}0.0308}\) & \(\mathbf{0.1621{\pm}0.0308}\) & 0.8112 & 1.0000 & \textbf{0.0} \\
\midrule
Trained LoRA &Equal & \(0.2284{\pm}0.0441\) & \(0.2284{\pm}0.0441\) & -0.1715 & 0.5000 & 0.0 \\
Trained LoRA &Volume & \(0.2284{\pm}0.0441\) & \(0.2284{\pm}0.0441\) & -0.1715 & 0.2500 & 0.0 \\
Trained LoRA &Leave-one-out & \(0.2284{\pm}0.0441\) & \(0.2284{\pm}0.0441\) & \textbf{1.0000} & 0.9166 & 0.0 \\
\rowcolor{keyrow} Trained LoRA &\method{} & \(\mathbf{0.2555{\pm}0.0063}\) & \(\mathbf{0.2856{\pm}0.0063}\) & 0.8418 & \textbf{1.0000} & \textbf{0.0} \\
\bottomrule
\end{tabular}
\end{table*}

\subsection{De-Circularized Held-Out Serving}

Table~\ref{tab:generator-backed-rag} is the primary serving experiment, and it de-circularizes evaluation by construction: after selecting coalitions on a validation card, the paid coalition is served on a disjoint test card with a real seq2seq reader. Under a high-volume held-out prompt-injection poisoner, \method{} improves downstream accuracy by 7.5--8.1 percentage points over Volume, leave-one-out, and FL-Shapley, with paired 95\% CIs that remain positive (\(0.0812{\pm}0.0477\) vs FL-Shapley, \(0.0813{\pm}0.0372\) vs leave-one-out, and \(0.0750{\pm}0.0372\) vs volume). The weaker retrieve--read proxy, moved to Appendix Table~S21, has the same direction of effect but is near the three-class chance floor: \method{} has the highest mean accuracy and macro-F1, but several CIs overlap.

\subsection{Market Utility, Ordered Credit, and Adapter Evidence}

The FEVER RAG market illustrates the failure \method{} targets: Equal, volume, retrieval-similarity, leave-one-out, FL-Shapley, and Shapley-UCB all select on average two strategic clients under the configured budget, while \method{} selects none and reaches the honest-client utility reference (harmful-client AUC is only a constructed-attack diagnostic; the robustness case rests on held-out serving, strategic counts, and ablations). The second load-bearing result is that ordered \valuator{} changes rank and selected-set evidence enough to be a contract choice, and the controlled study of Table~S15 shows the change is serving-neutral---a payment-fairness choice, not an accuracy lever. On the trained low-rank LoRA validation, \method{} obtains the best utility and risk-adjusted utility while retaining the rare adapter and selecting no strategic clients.

\subsection{De-Circularized Serving at Scale}
\label{sec:scale-heldout}

To establish serving quality at scale independently of the market's own objective, we run a larger, fully de-circularized end-to-end test. We build FEVER retrieval markets at 50, 100, and 200 clients, each with a rare specialist and roughly 12\% strategic clients: a high-volume poison corpus that is validation-useful but injects flipped-label copies of the held-out gold evidence, plus duplicate and sybil identities. Every baseline selects a budget-feasible coalition on a validation card; the selected coalition is then served on a \emph{disjoint} held-out test card and scored only on downstream task accuracy, macro-F1, and evidence exact match---metrics that use neither attack labels nor the risk-adjusted-utility formula. Table~\ref{tab:scale-heldout} reports held-out accuracy. Volume and leave-one-out buy the high-volume poison (6, 12, and 24 poison clients at the three scales) and collapse to $0.07$--$0.13$ accuracy; sampled FL-Shapley is similarly harmed. \method{} selects zero strategic clients at every scale and is the best or statistically tied-for-best method---its interval overlaps Equal at 50 clients and Shapley-UCB at 100, and it is highest at 200---while beating the strongest \emph{standard valuation} baseline (Volume, leave-one-out, or FL-Shapley) by $+0.26$, $+0.44$, and $+0.50$ as the market grows. Because the score is a held-out task metric, this is a non-circular demonstration that \method{}'s risk-aware selection converts strategic-robustness signals into serving quality, and that the margin widens with market size. Two comparisons place the margin in context: Shapley-UCB, which also discounts uncertain sellers, is the closest competitor, and simple Equal-weighting performs respectably because the cost model already embeds a risk term. The 200-client markets run in about $27$~s in the prototype.

\begin{table}[t]
\centering
\caption{De-circularized held-out serving at scale. Test-card accuracy (mean \(\pm\) 95\% CI) for coalitions selected on a separate validation card; no attack labels or risk-adjusted utility enter the score.}
\label{tab:scale-heldout}
\scriptsize
\begin{tabular}{lccc}
\toprule
\rowcolor{headrow} Method & 50 clients & 100 clients & 200 clients \\
\midrule
Equal & \(\mathbf{0.475{\pm}0.060}\) & \(0.500{\pm}0.074\) & \(0.517{\pm}0.188\) \\
Volume & \(0.070{\pm}0.018\) & \(0.090{\pm}0.033\) & \(0.100{\pm}0.075\) \\
Leave-one-out & \(0.085{\pm}0.040\) & \(0.100{\pm}0.041\) & \(0.133{\pm}0.091\) \\
FL-Shapley & \(0.205{\pm}0.185\) & \(0.095{\pm}0.042\) & \(0.108{\pm}0.071\) \\
Retrieval similarity & \(0.420{\pm}0.048\) & \(0.465{\pm}0.070\) & \(0.483{\pm}0.145\) \\
Shapley-UCB & \(0.455{\pm}0.059\) & \(\mathbf{0.545{\pm}0.059}\) & \(0.592{\pm}0.114\) \\
\rowcolor{keyrow} \method{} & \(0.465{\pm}0.091\) & \(0.540{\pm}0.045\) & \(\mathbf{0.633{\pm}0.114}\) \\
\bottomrule
\end{tabular}
\end{table}

\section{Discussion}
\label{sec:discussion}

\subsection{Value Drift}

Contribution values depend on the base model, prompt, retrieval stack, adapter routing, evaluation distribution, and safety rules, so the value index must support revaluation triggers: a client valuable for one base model may be redundant or harmful for another.

\subsection{Limitations and Scope}

\method{} occupies a specific operating point. FL-Shapley and Shapley-UCB can exceed it on raw utility because they optimize marginal validation utility directly, whereas \method{} prices contribution under duplicate, sybil, poison, privacy, and audit constraints, preserves scarce clients when the validation slice supports them, and produces a ledger that can be disputed and revalued. Its objective is therefore a practical balance between utility, robustness, participation fairness, and auditability rather than immediate accuracy alone. Two scope notes follow. The adapter track uses a trained, cached HuggingFace/PEFT LoRA path; full multi-adapter PEFT valuation is a scale-up experiment rather than the default. And because raw artifacts are privacy-constrained, \method{} provides approximate manipulation resistance---lowering expected gains from common attacks, exposing uncertainty, and making high-risk payments auditable---rather than a dominant-strategy truthfulness guarantee.

\subsection{Privacy Limits}

The framework reduces raw data exposure but does not eliminate privacy risk. The prototype provides bounded DP aggregate releases with composition accounting and third-party-verifiable (Ed25519) signed evidence whose schema maps field-for-field to hardware TEE attestation documents (Table~S13); payment values are protected through secure or attested evaluation and ledger access control rather than per-value DP. Table~\ref{tab:securityprimitives} separates the mandatory primitives (signed contract cards, payload hashes, bounded DP releases, hidden-probe separation, immutable logs, per-round accounting) from those needed for stronger confidentiality (secure aggregation and hardware-backed attestation), and states what fails without each. Without hardware attestation or secure aggregation, the prototype is a reproducible evidence path, not an end-to-end privacy proof.

\begin{table}[t]
\centering
\caption{Security and privacy primitives.}
\label{tab:securityprimitives}
\scriptsize
\begin{tabular}{@{}p{0.23\linewidth}p{0.18\linewidth}p{0.39\linewidth}@{}}
\toprule
\rowcolor{headrow} Primitive & Status & What fails if missing \\
\midrule
Signed contract card & Mandatory & Rare slices, weights, and probes can be retrofitted. \\
Payload hashes & Mandatory & Duplicate and dispute evidence is not replayable. \\
Bounded DP aggregate releases & Mandatory for public aggregates & Published statistics can leak sensitive values. \\
Per-round privacy accounting & Mandatory & \((\epsilon,\delta)\) spend cannot be audited. \\
Hidden audit probe separation & Mandatory & Strategic clients can overfit public probes. \\
Secure aggregation & Mandatory for private update sketches & Operator can inspect individual updates. \\
Hardware TEE attestation & Mandatory for raw-payload evaluation & Signing key has no hardware root of trust (evidence stays third-party verifiable in software). \\
Watermark checks & Optional defense & Adaptive mimicry becomes easier. \\
\bottomrule
\end{tabular}
\end{table}

\section{Conclusion}
\label{sec:conclusion}

This paper introduced \method{}, an auditable, risk-adjusted data-market framework for federated foundation-model adaptation. The core idea is to value and pay for heterogeneous foundation-model adaptation artifacts under privacy, non-IID heterogeneity, strategic behavior, and system constraints. \method{} combines a contribution registry, secure evaluation sandbox, unordered or pipeline-ordered \valuator{} valuation engine, uncertainty-aware payment rule, market ledger, and dispute auditor. Experiments on FEVER retrieval, held-out RAG serving, low-rank PyTorch adapters, and cached HF/PEFT LoRA validation show that \method{} preserves held-out task quality under adversarial participation---improving downstream accuracy by 7.5--8.1 points over standard valuation baselines while selecting zero strategic clients---and provides a concrete market layer for ordered, risk-adjusted, rare-client-preserving, and audit-ready FedFM adaptation contribution valuation.

\section*{Data and Code Availability}
The source code, benchmark harness (\bench{}), datasets, and scripts required to reproduce all experiments, tables, and figures in this paper are available at \url{https://anonymous.4open.science/r/FedMark-FM-3A89}.


\clearpage
\setcounter{section}{0}
\setcounter{table}{0}
\setcounter{figure}{0}
\setcounter{equation}{0}
\renewcommand{\thesection}{S\arabic{section}}
\renewcommand{\thetable}{S\arabic{table}}
\renewcommand{\thefigure}{S\arabic{figure}}
\renewcommand{\theequation}{S\arabic{equation}}

\section*{Supplementary Material}
This supplement contains the appendices referenced from the main paper; its sections, tables, figures, and equations are numbered with an ``S'' prefix.

\section{Deferred Proofs and Additional Results}
\subsection{Theoretical Properties}

We establish the following properties under assumptions that define their scope. Utilities are bounded in \([a,b]\). The surrogate marginal prediction error is at most \(\epsilon_{\mathrm{sur}}\) in expectation. Privacy-preserving sketches introduce at most \(\epsilon_{\mathrm{priv}}\) utility distortion. Client metadata and provenance commitments are verifiable and cannot be forged. Duplicate clustering identifies a same-source split with recall at least \(r_{\mathrm{dup}}\). A hidden audit detects a poisoned contribution with probability \(q\) and applies audit penalty \(\kappa_{\mathrm{aud}}\). A strategic deviation can increase apparent value by at most \(g_i\) before market penalties. For sybil analysis, \(\epsilon_{\mathrm{split}}\) is the maximum extra apparent value created by splitting a fixed portfolio, and \(\delta\) is the duplicate penalty applied to each detected extra identity in the same redundancy cluster.

Using the notation above, final payments under budget \(B\) are
\begin{equation}
p_i =
\begin{cases}
\widetilde{p}_i, & \sum_j \widetilde{p}_j \leq B,\\
B\widetilde{p}_i / \sum_j \widetilde{p}_j, & \sum_j \widetilde{p}_j > B.
\end{cases}
\end{equation}
Client utility is \(u_i=p_i-c_i^{act}-\ell_i\), where \(c_i^{act}\) is actual participation cost and \(\ell_i\) is privacy or manipulation loss.

\begin{proposition}[\valuator{} estimation error]
Let \(\bar{\Delta}_i\) be the empirical mean of \(m_i\) bounded direct marginal evaluations and let \(\widehat{\phi}_i\) replace some direct calls with a surrogate and DP/sketch releases. If utilities lie in \([a,b]\), the surrogate error is at most \(\epsilon_{\mathrm{sur}}\) with probability \(1-\delta_{\mathrm{sur}}\), and privacy/sketch distortion is at most \(\epsilon_{\mathrm{priv}}\) with probability \(1-\delta_{\mathrm{priv}}\), then with probability at least \(1-\delta-\delta_{\mathrm{sur}}-\delta_{\mathrm{priv}}\),
\[
|\widehat{\phi}_i-\phi_i|
\leq
(b-a)\sqrt{\frac{\log(2/\delta)}{2m_i}}
\;+\;\epsilon_{\mathrm{sur}}+\epsilon_{\mathrm{priv}} .
\]
\end{proposition}

\begin{IEEEproof}
Let \(\widetilde{\phi}_i\) be the empirical mean of true sampled marginals. Hoeffding's inequality gives the first term for \(|\widetilde{\phi}_i-\phi_i|\). The estimator differs from \(\widetilde{\phi}_i\) only through surrogate replacement and privacy/sketch perturbation. A union bound over the three events and the triangle inequality yield the additive decomposition.
\end{IEEEproof}

We calibrate the LCB intervals empirically. We use split conformal calibration: sampled marginal estimates are computed on a proper valuation split, absolute residuals are computed on a calibration split, and the lower predictive bound for a test client is \(\widehat{\phi}_i-\widehat{q}_{1-\alpha}\), where \(\widehat{q}_{1-\alpha}\) is the finite-sample conformal residual quantile. Under exchangeability of calibration and test residuals within a contract stratum, this gives marginal \(1-\alpha\) lower-bound coverage for the target value used by that stratum. Exchangeability is only approximate under non-IID coalitions, so we report coverage rather than assuming it. On small FEVER submarkets with at most ten clients, we compute exact Shapley values exhaustively; naive sampled intervals cover 0.3333 of exact values, while split conformal intervals cover 1.0 with mean width 0.0141. On larger markets where exact Shapley is infeasible, the same procedure is reported against leave-one-client-out and sampled-Shapley surrogates as a diagnostic, not as a formal proof.

\begin{theorem}[Selection correctness under value gaps]
Fix a budget \(B\), a deterministic budgeted selector \(A(\cdot)\), and a downstream test utility \(U_{\mathrm{test}}\) that is independent of attack labels and payment penalties. Let \(S^\star=A(\phi)\) denote the budget-feasible downstream-oracle set under true client values \(\phi_i\), and let \(\widehat{S}=A(\widehat{\phi})\) be the paid set obtained from estimated values. Suppose \(S^\star\) has margin
\[
\Delta_A =
\min_{i\in S^\star,j\notin S^\star}
\left(\frac{\phi_i}{c_i}-\frac{\phi_j}{c_j}\right) > 0
\]
with respect to the value-density ordering used by \(A\), and costs satisfy \(c_i\leq c_{\max}\). If every client has
\[
m_i \geq
\frac{2(b-a)^2}{(\Delta_A c_i/2-\epsilon_{\mathrm{sur}}-\epsilon_{\mathrm{priv}})^2}
\log\frac{2n}{\delta}
\]
direct-equivalent audits and \(\Delta_A c_i/2>\epsilon_{\mathrm{sur}}+\epsilon_{\mathrm{priv}}\), then
\[
\Pr[\widehat{S}\neq S^\star]\leq\delta .
\]
Consequently, the operator regret against the downstream oracle satisfies
\[
U_{\mathrm{test}}(S^\star)-U_{\mathrm{test}}(\widehat{S})
= 0
\]
on the same high-probability event, and is at most the range of \(U_{\mathrm{test}}\) otherwise.
\end{theorem}

\begin{IEEEproof}
By Proposition~1 and a union bound over \(n\) clients, all value estimates are within \(\Delta_A c_i/2\) of their true values with probability at least \(1-\delta\). On that event, every selected client's value density remains above every unselected client's value density, so the deterministic selector returns \(S^\star\). The regret statement follows because the selected set is identical on the high-probability event; outside the event, regret is trivially bounded by the utility range.
\end{IEEEproof}

\noindent\textbf{Remark 1 (budget feasibility).}
Budget feasibility is an arithmetic invariant of the scaled payment rule. If \(\sum_i\widetilde{p}_i\leq B\), payments are unchanged; otherwise \(p_i=B\widetilde{p}_i/\sum_j\widetilde{p}_j\), so \(\sum_i p_i=B\).

\noindent\textbf{Remark 2 (payment monotonicity).}
For two clients with equal cost, privacy, risk, uncertainty, and scarcity terms, a larger \(LCB_i\) gives a weakly larger raw payment. Budget scaling multiplies all positive raw payments by the same nonnegative factor, preserving the order.

\begin{proposition}[Approximate incentive compatibility]
Suppose a deviation can increase client \(i\)'s apparent payment by at most \(g_i\), but triggers expected audit penalty \(q\kappa_{\mathrm{aud}}\), risk discount \(\eta r_i\), and manipulation loss \(\ell_i^{man}\). Honest reporting is \(\epsilon_i\)-dominant with
\[
\epsilon_i = \max(0, g_i-q\kappa_{\mathrm{aud}}-\eta r_i-\ell_i^{man}).
\]
In particular, the deviation is unprofitable when \(q\kappa_{\mathrm{aud}}+\eta r_i+\ell_i^{man}\geq g_i\).
\end{proposition}

\begin{IEEEproof}
The deviation's utility gain is at most its apparent payment gain minus expected audit penalty, risk discount, and manipulation loss. Taking the positive part gives \(\epsilon_i\).
\end{IEEEproof}

\begin{proposition}[Sybil splitting unprofitability]
If a portfolio is split into \(k\) same-cluster identities and the aggregate split value can exceed the honest unsplit value by at most \(\epsilon_{\mathrm{split}}\), then
\[
Gain_{\mathrm{sybil}} \leq \epsilon_{\mathrm{split}} - r_{\mathrm{dup}}(k-1)\delta .
\]
Sybil splitting is unprofitable whenever \(r_{\mathrm{dup}}(k-1)\delta \geq \epsilon_{\mathrm{split}}\).
\end{proposition}

\begin{IEEEproof}
The split can add at most \(\epsilon_{\mathrm{split}}\) apparent value. With probability or recall \(r_{\mathrm{dup}}\), each extra identity after the first receives duplicate penalty \(\delta\). Subtracting expected penalties gives the bound.
\end{IEEEproof}

\begin{proposition}[Audit-adjusted poisoning profitability]
A poisoned contribution with apparent gain \(g\), audit detection probability \(q\), audit penalty \(\kappa_{\mathrm{aud}}\), and risk score \(r_i\) has expected extra profit at most
\[
\mathbb{E}[profit_{\mathrm{poison}}] \leq g-q\kappa_{\mathrm{aud}}-\eta r_i.
\]
Poisoning is unprofitable whenever \(q\kappa_{\mathrm{aud}}+\eta r_i\geq g\).
\end{proposition}

\begin{IEEEproof}
The apparent gain is \(g\). The expected audit penalty is \(q\kappa_{\mathrm{aud}}\), and the deterministic market risk discount is \(\eta r_i\). Summing the terms yields the inequality.
\end{IEEEproof}

\begin{proposition}[Individual rationality]
For an honest client with \(\ell_i=0\), if \(p_i\geq c_i^{act}\), then participation is individually rational. A sufficient pre-scaling condition is
\[
LCB_i+\rho s_i \geq \beta c_i+\gamma \epsilon_i^{priv}+c_i^{act}.
\]
\end{proposition}

\begin{IEEEproof}
The first statement follows from \(u_i=p_i-c_i^{act}\geq 0\). The sufficient condition makes the raw payment at least actual cost when risk is zero; budget scaling preserves non-negativity, and individual rationality can be enforced by rejecting clients whose scaled payment falls below reserve cost.
\end{IEEEproof}

\subsection{Extended Validation}

We also run a larger FEVER validation with 28 clients, including 24 honest clients and four strategic clients. Table~\ref{tab:extended} shows that FL-Shapley obtains the best raw utility in this larger setting, while \method{} is close in raw utility and does not force rare-client selection when the rare slice is not utility-improving. Consistent with its design, \method{} optimizes a different operating point from a pure accuracy maximizer: raw utility can favor FL-Shapley or Shapley-UCB, while \method{} targets risk-adjusted and audit-ready market value. Figure~\ref{fig:riskadjusted} plots the risk-adjusted utilities for the large FEVER RAG market and the trained LoRA coalition track, where \method{} reaches 0.65 versus 0.22--0.29 for the baselines.

Finally, we run actual HuggingFace/PEFT LoRA experiments using the pretrained public encoder \texttt{distilbert-base-uncased}, replacing the earlier random tiny test fixture. A single-adapter smoke run completes locally with 630,532 trainable LoRA parameters out of 67,587,080 total parameters (0.9329\%), takes 5.6732 seconds for training and 1.6443 seconds for inference, and reaches 0.4688 AG News accuracy on a 64-example smoke test. We also run a PEFT coalition market: six client adapters are trained separately, and coalitions are evaluated by averaging adapter logits. In this coalition track, \method{} selects no strategic adapters, retains the rare adapter, and reaches 0.6462 risk-adjusted utility versus 0.2200 for equal and volume selection and 0.2859 for leave-one-out. These results are not competitive classification benchmarks; they verify a real pretrained-transformer PEFT coalition path.

\begin{figure}[t]
\centering
\includegraphics[width=\linewidth]{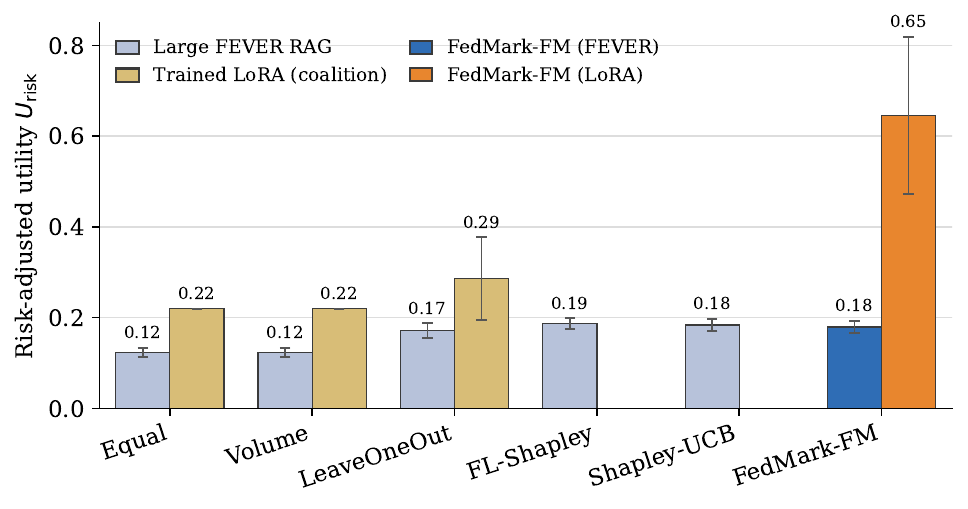}
\caption{Risk-adjusted utility \(U_{\mathrm{risk}}\) for the large FEVER RAG market and the trained LoRA coalition track; higher is better. \method{} is highlighted in each group.}
\label{fig:riskadjusted}
\end{figure}

\begin{table*}[t]
\centering
\caption{Extended validation. Large FEVER uses 28 clients. PEFT validates a trained, cached LoRA path.}
\label{tab:extended}
\scriptsize
\begin{tabular}{lccc}
\toprule
\rowcolor{headrow} Method/track & Utility & R-utility & AUC \\
\midrule
Large FEVER FL-Shapley & \(\mathbf{0.1869{\pm}0.0118}\) & \(\mathbf{0.1869{\pm}0.0118}\) & 1.0000 \\
Large FEVER Shapley-UCB & \(0.1837{\pm}0.0138\) & \(0.1837{\pm}0.0138\) & 0.6573 \\
\rowcolor{keyrow} Large FEVER \method{} & \(0.1792{\pm}0.0137\) & \(0.1792{\pm}0.0137\) & 1.0000 \\
Large FEVER Leave-one-out & \(0.1712{\pm}0.0163\) & \(0.1712{\pm}0.0163\) & 1.0000 \\
Large FEVER Volume & \(0.0938{\pm}0.0098\) & \(0.1238{\pm}0.0098\) & 0.7500 \\
\midrule
HF/PEFT LoRA smoke & 0.4587 & N/A & N/A \\
HF/PEFT LoRA coalition & \(0.6162{\pm}0.1732\) & \(\mathbf{0.6462{\pm}0.1732}\) & N/A \\
\bottomrule
\end{tabular}
\end{table*}

\subsection{Scalability}

Table~\ref{tab:scale} reports the original RAG valuation runtime as the number of clients and coalition samples increase. The script counts two utility calls per sampled marginal. To ground the data-systems claim in measurement rather than projection, Table~\ref{tab:measured-scale} adds measured FEVER markets at approximately 200, 500, and 1000 clients with a non-degenerate coverage utility. These larger rows use fewer coalition samples because they validate wall-clock scale and selection quality, not to replace the smaller high-fidelity RAG evaluator. Client counts include four strategic stress clients; Volume uses no coalition utility calls, while FL-Shapley and \method{} share sampled-marginal calls. The table uses a tight selection budget with method-specific runtime and call accounting, so Volume and FL-Shapley do not coincide. Under this cheap coverage utility, FL-Shapley is the strongest raw selector at scale, while \method{} remains above Volume but gives up raw coverage at 1000 clients because its risk/duplication discounts are not tuned for the proxy. A formal cost decomposition \(T(n,m,k,h)=T_{\mathrm{init}}(n)+m\!\cdot\!(2h\,T_{\mathrm{ev}}(n)+(1\!-\!h)\,T_{\mathrm{su}}+T_{\mathrm{fe}}+T_{\mathrm{up}})+L\,T_{\mathrm{dp}}+R\,T_{\mathrm{sg}}\), a coarse fit of \(T_{\mathrm{step}}(n)\!\approx\!2.0\!\times\!10^{-3}\,n^{1.5}\) seconds for \(n\geq16\), a per-component breakdown, and a direct \(T_{\mathrm{ev}}(n)\) microbenchmark are given in the appendix \valuator{} Computational Cost subsection (Tables~\ref{tab:complexity-fit}, \ref{tab:msweep}, and~\ref{tab:tev-microbench}). The analytic model serves as an explanatory fit checked against measurements rather than as primary evidence.

\begin{table}[t]
\centering
\caption{Scalability of sampled valuation on the FEVER RAG track.}
\label{tab:scale}
\begin{tabular}{rrrr}
\toprule
\rowcolor{headrow} Clients & Samples & Runtime (s) & Calls \\
\midrule
8 & 6 & 0.6495 & 12 \\
8 & 16 & 1.1540 & 32 \\
16 & 6 & 0.8888 & 12 \\
16 & 16 & 1.5915 & 32 \\
32 & 6 & 1.8874 & 12 \\
32 & 16 & 5.0247 & 32 \\
64 & 6 & 6.2499 & 12 \\
64 & 16 & 15.1439 & 32 \\
\bottomrule
\end{tabular}
\end{table}

\begin{table}[t]
\centering
\caption{Measured FEVER market scale. Selection quality is the selected coalition's coverage utility divided by a greedy downstream oracle; higher is better.}
\label{tab:measured-scale}
\scriptsize
\begin{tabular}{llrrrr}
\toprule
\rowcolor{headrow} Clients & Method & Calls & Runtime (s) & Quality & Gap \\
\midrule
204 & Volume & 0 & 0.0001 & 0.6658 & 0.1154 \\
204 & FL-Shapley & 816 & 0.0092 & \textbf{0.9260} & \textbf{0.0256} \\
\rowcolor{keyrow} 204 & \method{} & 816 & 0.0092 & 0.9056 & 0.0326 \\
\midrule
504 & Volume & 0 & 0.0002 & 0.5968 & 0.1876 \\
504 & FL-Shapley & 1008 & 0.0256 & \textbf{0.8290} & \textbf{0.0794} \\
\rowcolor{keyrow} 504 & \method{} & 1008 & 0.0256 & 0.8170 & 0.0851 \\
\midrule
1004 & Volume & 0 & 0.0003 & 0.4517 & 0.2611 \\
1004 & FL-Shapley & 2008 & 0.0819 & \textbf{0.6331} & \textbf{0.1748} \\
\rowcolor{keyrow} 1004 & \method{} & 2008 & 0.0819 & 0.5260 & 0.2257 \\
\bottomrule
\end{tabular}
\end{table}

\subsection{Attack Tests}

Table~\ref{tab:attack} reports stress tests on the real-data track. The values are attack indicators: for duplicate and poison, lower is better; for non-IID specialist, higher is better. \method{} passes all tested attacks on both tracks.

\begin{table}[t]
\centering
\caption{Attack summary for the real-data track.}
\label{tab:attack}
\begin{tabular}{llrr}
\toprule
\rowcolor{headrow} Track & Attack & Mean value & Pass rate \\
\midrule
FEVER RAG & Duplicate & 0.0 & 1.0 \\
FEVER RAG & Poison & 0.0 & 1.0 \\
FEVER RAG & Sybil & 0.0 & 1.0 \\
FEVER RAG & Non-IID specialist & 1.0 & 1.0 \\
\bottomrule
\end{tabular}
\end{table}

\subsection{Consolidated Robustness and Sensitivity Checks}

We add targeted checks for the most fragile parts of the market layer. Instead of presenting each auxiliary CSV as a separate appendix table, Fig.~\ref{fig:grouped-evidence} groups roughly fifteen checks into four panels, reported in two consolidated tables. Table~\ref{tab:reviewresponse} covers calibration and defenses: naive sampled-marginal intervals are overconfident while exact-Shapley split-conformal calibration reaches full coverage on small FEVER submarkets; no single surrogate family dominates every artifact type; the hard-paraphrase duplicate threshold gives zero false positives at recall 0.9575; and the incentive-compatibility bound is zero under full defenses but positive (attacker profit 0.035--0.075) under a deliberately weakened profile, so the empirical bound tracks defense settings. Table~\ref{tab:reviewsensitivity} covers sensitivity and system behavior: budget, privacy-cost, DP-\(\epsilon\), and policy-drift sweeps, per-type LCB fairness, wall-clock overhead, multi-round drift, and VCG, procurement, and ordered-valuation baselines.

% \begin{figure}[t]
% \centering
% \resizebox{\linewidth}{!}{%
% \begin{tikzpicture}[
%   panel/.style={draw, rounded corners=1pt, align=center, font=\scriptsize, minimum width=2.55cm, minimum height=1.05cm},
%   arrow/.style={-{Latex[length=1.3mm]}, thick}
% ]
% \node[panel] (a) {Calibration\\CI, exact Shapley\\surrogate audits};
% \node[panel, right=0.35cm of a] (b) {Defenses\\duplicates, collusion\\IC bounds};
% \node[panel, below=0.35cm of a] (c) {Sensitivity\\privacy, budget\\policy drift};
% \node[panel, right=0.35cm of c] (d) {System\\scale, overhead\\multi-round};
% \node[draw, rounded corners=1pt, align=center, font=\scriptsize, minimum width=5.0cm, minimum height=0.65cm] (claim) at ($(c.south)!0.5!(d.south)+(0,-0.55cm)$) {auditable risk-adjusted market evidence};
% \draw[arrow] (a.west) -- ++(-0.25,0) |- (claim.west);
% \draw[arrow] (b.east) -- ++(0.25,0) |- (claim.east);
% \draw[arrow] (c.south)--(claim.north west);
% \draw[arrow] (d.south)--(claim.north east);
% \end{tikzpicture}
% }
% \caption{Grouped robustness evidence. The detailed per-check rows are compressed into four evidence panels instead of many isolated appendix tables.}
% \label{fig:grouped-evidence}
% \end{figure}

\begin{figure}
  \centering
  \includegraphics[width=1\linewidth]{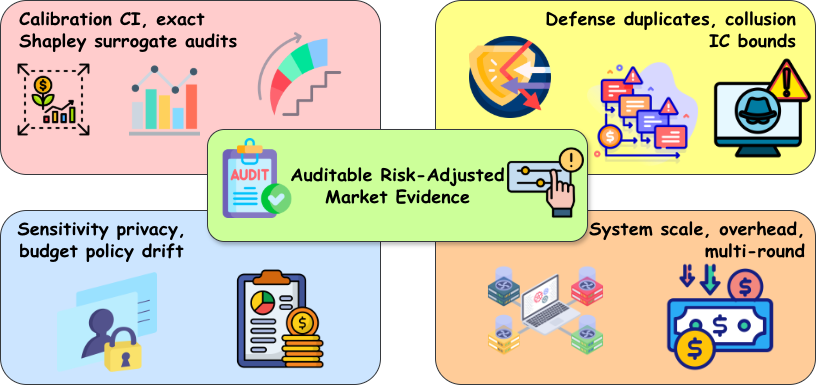}
  \caption{Consolidated robustness and sensitivity evidence, compressing roughly fifteen stress and sensitivity checks into four panels.}
  \label{fig:grouped-evidence}
\end{figure}

\begin{table*}[t]
\centering
\caption{Calibration, surrogate, and defense checks.}
\label{tab:reviewresponse}\label{tab:reviewdefense}
\scriptsize
\setlength{\tabcolsep}{2pt}
\begin{tabular}{@{}p{0.21\textwidth}p{0.23\textwidth}p{0.15\textwidth}p{0.14\textwidth}p{0.17\textwidth}@{}}
\toprule
\rowcolor{headrow} \multicolumn{5}{@{}l}{\emph{(a) Calibration and surrogates}} \\
\midrule
\rowcolor{keyyellow} Check & Setting & Error/coverage & Rank/utility & Reading \\
\midrule
CI calibration & FEVER, LOO diagnostic & Naive 0.2500; calibrated 0.8875 & Spearman 0.8629 & Raw intervals need calibration. \\
Exact-Shapley CI & FEVER \(n\leq10\) & Naive 0.3333; conformal 1.0000 & Width 0.0141 & LCB coverage tied to exact values. \\
Surrogate audit budget & FEVER 0\%\(\rightarrow\)100\% direct & MAE 0.0371\(\rightarrow\)0.0062 & Spearman 0.4971\(\rightarrow\)0.8349 & Direct audits matter near payments. \\
Surrogate family & FEVER RF, 50\% audit & MAE 0.0104 & Spearman 0.6415; R-util. 0.1447 & RF is strongest on retrieval. \\
Surrogate family & AG News prior, 50\% audit & MAE 0.0399 & Spearman 0.4549; R-util. 0.6792 & Class-volume prior is competitive. \\
Surrogate family & HF/PEFT smoke, 50\% audit & MAE 0.038--0.084 & R-util. 0.6462 vs 0.2200 & Learned surrogate beats volume prior. \\
Generator RAG & FLAN-T5 held-out reader & Acc. 0.3937 & Gap 0.075--0.081 & Real generator shows robust serving gap. \\
\midrule
\rowcolor{headrow} \multicolumn{5}{@{}l}{\emph{(b) Defenses}} \\
\midrule
\rowcolor{keyyellow} Check & Setting & Detection & Market outcome & Reading \\
\midrule
Semantic clustering & Threshold 0.55 & P/R/F1 1.0000/0.9975/0.9987 & FPR 0.0000 & Easy duplicates are caught. \\
Embedding duplicate & Combined threshold 0.50 & P/R/F1 1.0000/0.9889/0.9944 & FPR 0.0000 & Hash embeddings improve paraphrase recall. \\
Hard paraphrases & Threshold 0.65 & P/R 1.0000/0.9575 & FPR 0.0000 & Stricter threshold trades recall for precision. \\
Collusion audit & Cap 0.10 vs none & Colluders 0.0 vs 2.0 & R-util. \(-0.0176\) vs \(-0.1373\) & Group cap suppresses artificial complementarity. \\
Higher-order collusion & 3 clients, cap 0.08 & Colluders 0.0 & R-util. \(-0.0366\) & Three-client split is also suppressed. \\
Collusion flags & Threshold 0.25 & Precision 0.0747; recall 1.0000 & Legit-FP 0.0 & Flags trigger audit, not automatic penalty. \\
IC bounds & Full defenses & \(q=0.09\)--1.0 & \(\epsilon_{\mathrm{IC}}=0\); profit 0 & Implemented attacks unprofitable. \\
IC bounds & No cluster, \(\eta=0.10\) & \(q=0\), \(\kappa=0\) & \(\epsilon_{\mathrm{IC}}\), profit 0.035--0.075 & Bound has discriminating power. \\
Held-out attacks & Adaptive keyword & Poison 0.0; strategic 0.2 & Acc. 0.3857; regret 0.0514 & Served test metric is not circular. \\
Cost gaming & \(1\times\)--\(10\times\) cost plus duplicate & Attacker 0.0 & Strategic 0.0 & Cost inflation does not raise profit. \\
Contract-card swap & Weight/slice/probe swaps & Detected 1.0 & Payment preserved 1.0 & Commit-reveal catches post-hoc changes. \\
\bottomrule
\end{tabular}
\end{table*}

\begin{table*}[t]
\centering
\caption{Sensitivity, operating-point, system, and policy-stability checks.}
\label{tab:reviewsensitivity}\label{tab:reviewsystem}
\scriptsize
\setlength{\tabcolsep}{2pt}
\begin{tabular}{@{}p{0.20\textwidth}p{0.24\textwidth}p{0.15\textwidth}p{0.15\textwidth}p{0.17\textwidth}@{}}
\toprule
\rowcolor{headrow} \multicolumn{5}{@{}l}{\emph{(a) Sensitivity and operating point}} \\
\midrule
\rowcolor{keyyellow} Check & Setting & Utility/stability & Selection & Reading \\
\midrule
Contract weights & FEVER supports slice & R-util. 0.2463 & Strategic 0.0; rare 1.0 & Slice-aware cards matter. \\
Budget sweep & FEVER \(B=0.9\rightarrow2.2\) & Audit 0.95--0.93 & Selected 5--13; rare 0--0.67 & More budget changes retention. \\
Privacy-cost sweep & \(\gamma=0\rightarrow0.35\) & R-util. 0.1964--0.1965 & Strategic 0.0; rare 1.0 & Ordering stable in this grid. \\
DP epsilon & \(\epsilon=0.5\rightarrow16\) & Rank corr. \(-0.4637\)--\(-0.3176\) & 16 releases & Per-value DP noise harms payments. \\
Policy stability & RAG top1/top5 vs top3 & Spearman 0.8948/0.9594 & Strategic 0.0; rare 1.0 & Retrieval policy changes values. \\
Drift stability & Dense/slice policies & Spearman 0.7995/0.5614 & Strategic 0.0; rare 1.0 & Slice drift is quantitatively visible. \\
Adapter policies & Hierarchical/orthogonal proxies & R-util. 0.7169/0.6749 & Strategic 0.0; rare 1.0 & Adapter routing is a contract dimension. \\
LCB fairness & Adapter per-type LCB & R-util. 0.6863 & Unc. ratio 1.352 & Per-type calibration helps rare clients. \\
Alpha Pareto & FEVER \(\alpha_s,\alpha_r\) grid & Rank 4.33 if \(\alpha_r=0\); 1.67 if \(\alpha_r>0\) & Raw 0.1664 & Composite ranking is transparent. \\
\midrule
\rowcolor{headrow} \multicolumn{5}{@{}l}{\emph{(b) System and policy stability}} \\
\midrule
\rowcolor{keyyellow} Check & Setting & Runtime/calls & Outcome & Reading \\
\midrule
System overhead & 100 clients, 16 samples & 3200 calls; 3.4146s & 937.3 calls/s & Smoke-shard overhead is small. \\
Utility kernel & Fixed shard, 100 clients & 0.296 ms/call & 24 utility calls & Confirms fast 100-client row. \\
Utility kernel & Scale-sweep shard, 8 clients & 1.371 ms/call & 24 utility calls & Corpus growth explains slower fit. \\
50+ client scale & FEVER/AG News proxy & 324/200 calls & 1.50s/1.10s & Larger proxy run remains practical. \\
Measured scale & FEVER 204/504/1004 clients & 0--2008 calls & 0.0001--0.0819s & Real wall-clock rows confirm scale by measurement. \\
DP enclave overhead & 20 releases & 0.0045s & 0.226 ms/release & Software evidence overhead is tiny. \\
Multi-round & 6 rounds & Drift 0.0000--0.0306 & Strategic 0.0 & New rounds handle policy drift. \\
VCG toy & Observable 6-client market & VCG 0.0799; \method{} 0.0588 & VCG needs raw obs. & VCG wins when observability holds. \\
Procurement baseline & FEVER OPT-style & R-util. 0.2472 & Rare 0.0 vs \method{} rare 1.0 & Budget efficiency differs from auditability. \\
Ordered valuation & FEVER ADS-style & Spearman 0.7554 & Overlap 0.6667 & Pipeline order can change payments. \\
\bottomrule
\end{tabular}
\end{table*}

\subsection{Ablations}

Table~\ref{tab:ablation} ablates duplicate, risk, scarcity, and full market penalties. The largest degradation comes from removing risk penalties, which admits poisoned or sybil clients and sharply reduces risk-adjusted utility. Removing scarcity hurts rare-specialist selection in the adapter track. These ablations support the claim that the payment mechanism is not merely a Shapley scorer; the market-specific adjustments drive strategic robustness.

\begin{table}[t]
\centering
\caption{Ablation summary.}
\label{tab:ablation}
\scriptsize
\setlength{\tabcolsep}{4pt}
\begin{tabular}{llrr}
\toprule
\rowcolor{headrow} Track & Variant & R-utility & Strategic \\
\midrule
FEVER RAG & Full & \textbf{0.1843} & \textbf{0.0} \\
FEVER RAG & No risk penalty & -0.4086 & 3.0 \\
FEVER RAG & No market penalties & -0.3091 & 2.67 \\
\bottomrule
\end{tabular}
\end{table}

\begin{table*}[t]
\centering
\caption{Positioning \method{} against adjacent literature. Coverage in the middle columns is shown as Harvey balls (empty = none, quarter = limited/indirect, half = partial, full = yes/addressed).}
\label{tab:related-positioning}
\setlength{\tabcolsep}{2pt}
\scriptsize
\begin{tabular}{@{}p{0.17\textwidth}p{0.21\textwidth}>{\centering\arraybackslash}p{0.10\textwidth}>{\centering\arraybackslash}p{0.085\textwidth}>{\centering\arraybackslash}p{0.12\textwidth}p{0.28\textwidth}@{}}
\toprule
\rowcolor{headrow} Line of work & Contribution unit & FM artifacts & Incentives & Strategic robustness & Gap addressed by \method{} \\
\midrule
FL incentives~\cite{yang2023flshapleypareto,zhang2024shapleyucb} & Client data or update & \hb{0.25} & \hb{1} & \hb{0.5} & Does not model RAG, prompts, adapters, preference, and safety artifacts as traded products. \\
Shapley/data valuation~\cite{ghorbani2019datashapley,jia2019efficient,yoon2020dvrl} & Examples, sources, clients & \hb{0.25} & \hb{0.25} & \hb{0.25} & Estimates contribution, but usually lacks a payment ledger, audits, and attack-aware market controls. \\
Data markets and auctions~\cite{jiang2021flmarket,lu2024datameasurements} & Datasets or sellers & \hb{0} & \hb{1} & \hb{0.5} & Prices data access, but does not bind valuation to private FedFM evaluation pipelines. \\
Structure-aware Shapley~\cite{zheng2024ads,zheng2025ads,chi2024pcwinter} & Ordered data groups & \hb{0.5} & \hb{0} & \hb{0.25} & Respects pipeline order, but does not define FedFM registry, payment, and audit machinery. \\
FedFM systems~\cite{fan2025fedfmchallenges,zhang2025fmincentive} & Private model/data collaboration & \hb{1} & \hb{0} & \hb{0} & Identifies incentives as a challenge, but leaves the concrete market mechanism unspecified. \\
Secure/incentive FL systems~\cite{addison2024cfedrag,yan2026fweb3} & Updates or RAG context & \hb{0.5} & \hb{1} & \hb{0.5} & Provides confidentiality or settlement substrate, but not typed FM artifact markets. \\
FM adaptation valuation~\cite{ye2025documentvaluation,han2025ragauctions} & Prompts, RAG passages, adapters & \hb{1} & \hb{0} & \hb{0.25} & Values artifacts for model improvement without federated sellers, payments, and dispute evidence. \\
\rowcolor{keyrow} \method{} & Typed private artifacts & \hb{1} & \hb{1} & \hb{1} & Integrates scalable valuation, budgeted payments, risk controls, and auditable system state. \\
\bottomrule
\end{tabular}
\end{table*}

\section{Registry, Ledger, and Experiment Schema}

Table~\ref{tab:registry-schema} gives the minimal registry fields used by the prototype. A deployment can add organization-specific compliance fields, but these fields are sufficient to reproduce valuation, payment, and dispute checks.

\begin{table}[t]
\centering
\caption{Contribution registry schema.}
\label{tab:registry-schema}
\scriptsize
\begin{tabular}{@{}p{0.33\linewidth}p{0.53\linewidth}@{}}
\toprule
\rowcolor{headrow} Field & Meaning \\
\midrule
client\_id & Stable seller identity or verified organization handle \\
artifact\_id & Unique artifact identifier bound to provenance hashes \\
artifact\_type & Retrieval, adapter, prompt, demonstration, preference, safety, or update sketch \\
domain\_tags & Declared and inferred task/domain strata for sampling \\
provenance\_hash & Signed hash or commitment for duplicate and dispute checks \\
privacy\_budget & Declared evaluation budget or privacy policy handle \\
declared\_cost & Claimed compute, curation, labeling, or serving cost \\
risk\_flags & Poison, duplicate, sybil, policy, or uncertainty indicators \\
endpoint\_ref & Federated retrieval, adapter, local evaluator, or certificate endpoint \\
license\_policy & Use constraints, retention policy, and payment eligibility \\
\bottomrule
\end{tabular}
\end{table}

The ledger stores the same round id, client id, estimated value, uncertainty, LCB, penalties, scarcity bonus, raw/scaled payment, evidence hashes, and dispute state so a payment can be replayed after model or policy changes. Table~\ref{tab:experiment-configs} summarizes the configuration of each experiment track.

\begin{table*}[t]
\centering
\caption{Experiment configuration card.}
\label{tab:experiment-configs}
\setlength{\tabcolsep}{2pt}
\scriptsize
\begin{tabular}{@{}p{0.17\textwidth}p{0.09\textwidth}p{0.12\textwidth}p{0.24\textwidth}p{0.29\textwidth}@{}}
\toprule
\rowcolor{headrow} Track & Seeds & Clients & Artifacts & Utility and baselines \\
\midrule
FEVER RAG & 3 & 18 & Claim evidence passages with duplicate, rare, and poison clients & Retrieval accuracy minus risk; equal, volume, leave-one-out, FL-Shapley, retrieval similarity, Shapley-UCB, \method{} \\
Trained low-rank adapter & 3 & 12 & Lightweight PyTorch low-rank adapters & Held-out classification accuracy and robustness checks \\
Large FEVER RAG & 10 & 28 & More clients and strategic variants & Multiseed RAG utility, risk-adjusted utility, AUC, and scalability counters \\
HF/PEFT LoRA smoke & 1 & 1 & Cached pretrained DistilBERT LoRA path & Training seconds, inference seconds, trainable parameters, total parameters, and LoRA percentage \\
HF/PEFT LoRA coalition & 2 & 4 & Real PEFT client adapters and coalition evaluations & Coalition utility with equal, volume, leave-one-out, and \method{} payment selections \\
\bottomrule
\end{tabular}
\end{table*}

\section{Notation and Contract Card}

Table~\ref{tab:notation} collects the notation used across valuation, payment, and theory. This table is intended to make the mechanism auditable: every payment field in the ledger maps to a symbol in the paper.

\begin{table}[t]
\centering
\caption{Notation summary.}
\label{tab:notation}
\scriptsize
\begin{tabular}{@{}p{0.21\linewidth}p{0.66\linewidth}@{}}
\toprule
\rowcolor{headrow} Symbol & Meaning \\
\midrule
\(N\) & Set of participating clients \\
\(\Pi_i,\Pi(C)\) & Client \(i\)'s private portfolio and coalition portfolio union \\
\(C\) & Coalition of clients evaluated by the sandbox \\
\(U(C)\) & Contract utility of coalition \(C\) \\
\(\phi_i\) & Ideal Shapley-style value of client \(i\) \\
\(\widehat{\phi}_i\) & Estimated client value returned by \valuator{} \\
\(\sigma_i\) & Estimation uncertainty for client \(i\) \\
\(LCB_i\) & Lower confidence bound \(\widehat{\phi}_i-\lambda\sigma_i\) \\
\(\lambda\) & Uncertainty discount; default value \(0.75\) \\
\(\beta\) & Cost-penalty coefficient; default value \(0.28\) \\
\(\gamma\) & Conversion from privacy/evaluation spend to payment penalty; default value \(0.20\) \\
\(\eta\) & Manipulation-risk penalty; default value \(0.75\) \\
\(\rho\) & Scarcity bonus weight; default value \(0.25\) \\
\(d_i,r_i\) & Duplicate risk and manipulation risk \\
\(\Phi(C)\) & Privacy-budget consumption or leakage term in utility \\
\(\kappa_{\mathrm{aud}}\) & Audit penalty applied after verified manipulation \\
\(c_i,\epsilon_i^{priv},s_i\) & Verified cost, privacy budget, and scarcity score \\
\(\widetilde{p}_i,p_i\) & Raw and budget-scaled payments \\
\(B\) & Market budget for the round \\
\bottomrule
\end{tabular}
\end{table}

A market round is defined by an immutable contract card containing the round id, base model, allowed artifact types, utility weights, validation slices, budget, valuation budget, payment parameters, audit policy, and release policy.

\section{Metric Definitions}

For a selected client set \(S\), the reported utility is the same contract utility used for market selection. Risk-adjusted utility adds a market-level robustness adjustment:
\begin{align*}
U_{\mathrm{risk}}(S) ={}& U(S)
-\alpha_{\mathrm{strat}} |S\cap S_{\mathrm{strategic}}|\\
&+\alpha_{\mathrm{rare}}\mathbf{1}\{S\cap S_{\mathrm{rare}}\neq\emptyset\}.
\end{align*}
The coefficients are fixed by the experiment card. Harmful-client AUC treats poisoned, sybil, and duplicate clients as positives and ranks clients by risk score. Duplicate discount is the ratio by which a duplicate client's payment is reduced relative to its unsuppressed score. Utility ratio divides a method's utility by the oracle subset utility under the same budget. Strategic selected is the number of selected clients with attack labels. Rare selected indicates whether at least one rare-domain specialist is retained.

\section{Stress-Test Protocol}

The stress tests instantiate sybil splitting, duplicate submission, poisoned contribution, privacy gaming, cost inflation, non-IID specialists, and collusion. Each test records selected strategic clients, attack profit, duplicate discount, rare retention, or reserve violation. The goal is executable measurement of common market failures, not immunity to arbitrary adaptive attacks.

\section{Contribution Interface Details}

The registry separates private payloads from market-visible metadata. The payload may remain at the client, inside a secure enclave, or behind a federated endpoint. The market-visible portion must be sufficient to schedule evaluations, detect obvious duplicates, compute privacy and cost adjustments, and bind later disputes to the same artifact. Table~\ref{tab:artifact-interface} details the required, optional, and private fields for each artifact type.

\begin{table*}[t]
\centering
\caption{Artifact-specific contribution interface. Required fields are visible to the market; private fields remain behind the endpoint or are exposed only through secure evaluation.}
\label{tab:artifact-interface}
\scriptsize
\setlength{\tabcolsep}{2pt}
\begin{tabular}{@{}p{0.14\textwidth}p{0.24\textwidth}p{0.27\textwidth}p{0.25\textwidth}@{}}
\toprule
\rowcolor{headrow} Artifact & Required market-visible fields & Optional sketches/certificates & Private payload \\
\midrule
Retrieval corpus & Domain tags, language, passage count, provenance hash, endpoint reference, license policy & MinHash, embedding centroids, BM25 term sketch, coverage histogram, duplicate hashes & Passage text, document metadata, private access logs \\
LoRA adapter & Base-model id, rank, target modules, parameter count, adapter hash, calibration split id & Delta norms, routing hints, public calibration trace, safety certificate & Adapter weights when served behind a private loader \\
Prompt template & Task tags, token length, prompt hash, policy constraints & Robustness trace, paraphrase sensitivity, public/private score gap & Full prompt if proprietary \\
Demonstrations & Task tags, count, format, label-space description, source hash & Embedding centroid, class/domain histogram, annotator reliability & Demonstration text or labels \\
Preference data & Pair count, preference dimensions, annotator policy, source hash & Disagreement rate, category histogram, reward-model certificate & Pairwise preference records \\
Safety data & Risk taxonomy, policy version, source hash, severity mix & Hidden-probe certificate, jailbreak family tags, red-team provenance & Full safety prompts and expected policies \\
Update sketch & Round id, model version, secure aggregation handle, update norm range & Clipped norm certificate, DP parameter certificate, validation delta & Raw gradient/update \\
\bottomrule
\end{tabular}
\end{table*}

The market operator should reject artifacts whose base-model id, license, endpoint policy, or provenance commitment is inconsistent with the round contract. This is a data-quality and governance check before any mechanism-design step is applied.

\subsection{Secure Evaluation Interfaces}

For retrieval corpora, the sandbox creates a temporary view or queries signed client endpoints, then computes hit rate, redundancy, latency, and risk without receiving the full corpus. For adapters, the contract fixes the base model, target modules, composition policy, and validation slices; our implementation includes a cached \texttt{distilbert-base-uncased} PEFT path. Prompts, demonstrations, preference data, and safety data are evaluated by contract probes, paraphrase variants, reward/preference probes, or hidden policy probes; public outputs are aggregate scores and evidence hashes unless a dispute requires controlled disclosure.

\section{Differential Privacy and Enclave Implementation}

Our implementation provides a concrete privacy layer. The DP module supports bounded mean releases with clipping, Laplace noise for pure DP, Gaussian noise for approximate DP, post-processing to the declared range, and a composition accountant. Gaussian releases are accounted with zCDP and can be related to RDP-style accounting \cite{bun2016concentrated,mironov2017renyi}:
\[
\rho = \sum_t \frac{\Delta_t^2}{2\sigma_t^2},
\quad
\epsilon(\delta)=\rho+2\sqrt{\rho\log(1/\delta)}.
\]
Laplace releases compose through sequential epsilon addition. Each release records \(n\), sensitivity, mechanism, noise scale, per-release privacy spend, and cumulative privacy spend.

The local secure-enclave backend enforces a policy before releasing any statistic. The policy specifies allowed functions, maximum epsilon, maximum delta, whether DP is required, and whether raw values may be released. The default policy forbids raw-value release. Every enclave output includes a manifest measurement, policy hash, input hash, output hash, cumulative DP spend, backend type, hardware-attestation flag, and a signature. The reproducible backend signs evidence with an Ed25519 key, so any auditor, buyer, or client can verify evidence with the public key alone; an HMAC, by contrast, can be forged by any party able to check it and is therefore not third-party verifiable. The genuine evidence verifies while any tamper to the released value or the freshness nonce invalidates the signature. The evidence schema mirrors the field layout of real TEE attestation documents (Table~\ref{tab:attestation-map}), so a hardware backend is swapped in by replacing only the signer and its certificate chain, leaving the market ledger unchanged. The one property the software backend cannot provide is a hardware root of trust for the signing key: on real hardware the \texttt{certificate}/\texttt{cabundle} fields carry a manufacturer-endorsed chain (AWS Nitro, Intel SGX/DCAP, AMD SEV-SNP), whereas the software backend self-signs them and reports \texttt{hardware\_attested=false}.

\begin{table}[t]
\centering
\caption{Attestation-evidence fields map to real TEE attestation documents; only the root of trust for the signing key differs in the software backend.}
\label{tab:attestation-map}
\scriptsize
\setlength{\tabcolsep}{3pt}
\begin{tabular}{@{}llll@{}}
\toprule
\rowcolor{headrow} \method{} evidence & AWS Nitro & Intel SGX/DCAP & AMD SEV-SNP \\
\midrule
\texttt{measurement} & PCR0 & MRENCLAVE & MEASUREMENT \\
\texttt{code\_hash} & PCR8 & MRSIGNER & ID\_KEY\_DIGEST \\
\texttt{public\_key} & \texttt{public\_key} & REPORTDATA & REPORT\_DATA \\
\texttt{nonce} & \texttt{nonce} & REPORTDATA & REPORT\_DATA \\
\texttt{output\_hash} & \texttt{user\_data} & REPORTDATA & REPORT\_DATA \\
\texttt{certificate} & \texttt{certificate} & PCK cert & VCEK cert \\
\texttt{signature} & COSE Sign1 & ECDSA quote & VCEK sig \\
\bottomrule
\end{tabular}
\end{table}

The prototype releases two DP aggregate statistics inside the software enclave: mean contribution quality and mean privacy budget. The privacy, enclave, and attestation components exercise bounded input, DP release, budget check, output hash, Ed25519-signed evidence, third-party verification, and ledger storage, writing signed evidence and an attestation document to disk. Table~\ref{tab:dpaccounting} lists the DP accounting settings used in these checks.

\begin{table}[t]
\centering
\caption{DP accounting settings used in our checks.}
\label{tab:dpaccounting}
\scriptsize
\setlength{\tabcolsep}{2pt}
\begin{tabular}{@{}p{0.19\linewidth}p{0.20\linewidth}p{0.22\linewidth}p{0.12\linewidth}p{0.12\linewidth}@{}}
\toprule
\rowcolor{headrow} Check & Mechanism & Per rel. \((\epsilon,\delta)\) & Clip & Releases \\
\midrule
Enclave mean & Gauss./Lap. & \((0.5,10^{-6})\) & \([0,1]\) & 1--20 \\
Value sweep & Laplace & \((0.5\text{--}16,10^{-6})\) & \([-1,1]\) & 16 \\
Privacy grid & Lap./zCDP proxy & \(\gamma\) grid & track-specific & 16--50 \\
\bottomrule
\end{tabular}
\end{table}

\enlargethispage{2\baselineskip}
The DP epsilon sweep shows the expected privacy-utility tension. With 16 bounded value releases on FEVER, strong noise at \(\epsilon=0.5\) sharply degrades rank correlation, and even \(\epsilon=16\) remains noisy in this tiny per-client release setting. We therefore do not recommend privatizing each individual value release independently in production: composition over many per-client, per-round value releases forces aggressive noise, which can invert payment ranks. The intended use is DP aggregate reporting plus secure or attested evaluation for payment-critical values. Figure~\ref{fig:dptradeoff} visualizes this privacy-utility tradeoff. Thus payments themselves are not claimed to be differentially private in the current prototype; their confidentiality comes from the secure-evaluation interface and evidence policy.

\section{\valuator{} Implementation and Sensitivity Notes}

\valuator{} samples marginal pairs \((C,i)\) by artifact type, validation slice, redundancy cluster, and risk stratum. The surrogate features include coalition size, artifact mix, slice coverage gaps, redundancy with \(C\), privacy/cost fields, drift, risk, scarcity, and a direct-evaluation indicator; retrieval and adapter tracks add sketch and adapter-fingerprint features. Direct audits are prioritized when uncertainty is high, a payment is near zero or the budget boundary, or a risk flag is present. The evaluation runner reports the detailed CSV rows; the compact reading is that direct audits reduce FEVER MAE from 0.0371 at 0\% audit to 0.0062 at 100\%, and that no single surrogate dominates every artifact type.

Budget, privacy, and policy sweeps are treated as calibration evidence rather than universal defaults. Increasing FEVER budget from 0.9 to 2.2 grows selected clients from 5 to 13 while preserving zero strategic selections in the tested grid. The privacy-cost sweep over \(\gamma\in\{0,0.05,0.10,0.20,0.35\}\) changes payment scores but not strategic selection in the current setup because risk and duplicate terms dominate. Operators should publish the Pareto frontier over raw utility, strategic selections, rare retention, audit rate, and dispute-candidate rate before settling a market round. Figure~\ref{fig:alphapareto} shows how the \method{} rank on FEVER depends on the operator-side risk-adjusted utility weights: the rank improves when rare retention is valued and weakens when \(\alpha_{\mathrm{rare}}=0\), making the operating-point tradeoff visible rather than implicit.

\textbf{Choosing the payment coefficients.} The weights \((\lambda,\beta,\gamma,\eta,\rho)\) are operator policy rather than data-fitted parameters, and we recommend selecting them from the Pareto frontier above: fix the uncertainty discount \(\lambda\) from the desired LCB coverage, set the cost and privacy conversions \(\beta,\gamma\) from verified unit prices, and choose the risk and scarcity weights \(\eta,\rho\) to hit a target strategic-selection and rare-retention rate. Selection is robust to the remaining slack because, as the sweeps show, the risk and duplicate terms dominate the ordering once \(\eta\) exceeds a small threshold, and the privacy-cost sweep leaves strategic selection unchanged. Duplicate and manipulation scores are audit \emph{triggers} backed by provenance commitments, not guarantees: token/trigram and embedding overlaps catch easy and paraphrased copies but can be evaded by strong adaptive paraphrase, so a deployment should add provenance attestation and watermark checks for defense-in-depth, and update-sketch tracks can plug a gradient-diversity sybil defense \cite{fung2018foolsgold} in behind the same evaluator interface.

\begin{figure}[t]
\centering
\includegraphics[width=\linewidth]{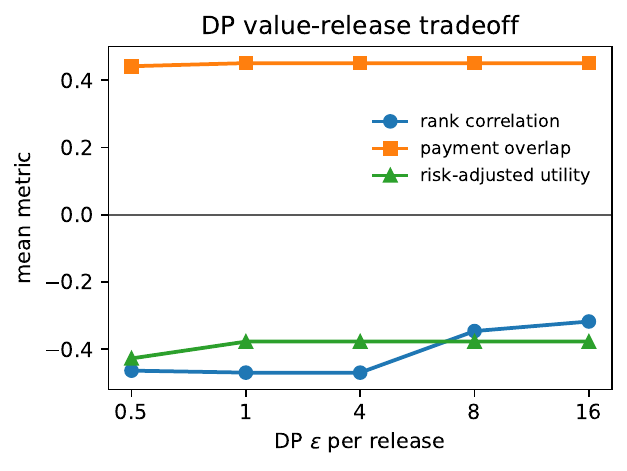}
\caption{DP value-release tradeoff on FEVER.}
\label{fig:dptradeoff}
\end{figure}

\begin{figure}[t]
\centering
\includegraphics[width=\linewidth]{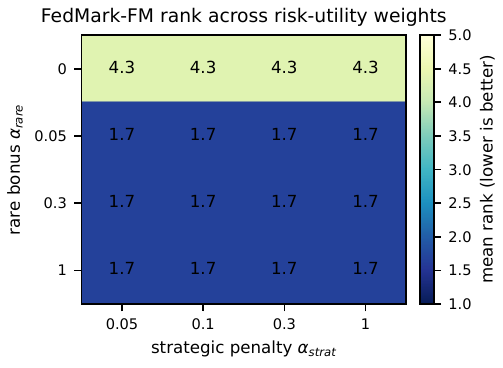}
\caption{Sensitivity of \method{}'s rank to the operator-side risk-adjusted utility weights on FEVER.}
\label{fig:alphapareto}
\end{figure}

Scarcity bonuses require pre-committed validation slices; post-hoc seller-proposed micro-slices are dispute evidence, not automatic bonuses. Duplicate defenses combine provenance hashes, lexical/embedding similarity, adapter delta fingerprints, target-module metadata, and hidden probes. The token/trigram tests are reproducible low-cost checks, not a complete adaptive-paraphrase defense.

\subsection{Ordered Pipeline Valuation}
\label{sec:ordered-appendix}

\valuator{} can use unordered or ADS-style ordered sampling. Ordered sampling restricts permutations to retrieval \(\rightarrow\) prompt/demonstration \(\rightarrow\) adapter \(\rightarrow\) preference/safety groups, preserving marginal-credit accounting while matching serving precedence. On FEVER, ordered S3Val has Spearman 0.7554 and selected-set overlap 0.6667 versus unordered S3Val. The contract card stores the group order so disputed payments can be replayed.

\textbf{Sensitivity to order misspecification.} Because credit depends on the declared group order, we test how much a wrong order matters. Within a group the value is permutation-invariant by within-layer symmetry, so only the cross-group precedence is consequential. Reversing the retrieval\(\prec\)reader precedence on the real-FEVER pipeline redistributes credit substantially (correct-versus-reversed Spearman $0.57$) and modestly lowers held-out serving accuracy (\(\Delta=-0.04\pm0.04\)), confirming that the serving order is a consequential contract parameter rather than a free choice. The operator therefore commits the order in the contract card and disputes replay payments under the recorded order; when the correct order is genuinely unknown, running unordered \valuator{} avoids imposing an unwarranted precedence.

\textbf{Is ordered valuation more correct, or only different?} Changed payments are not by themselves evidence that ordered credit is the right credit. We therefore ran a controlled held-out serving study on real FEVER data with two artifact groups that have a genuine serving precedence: retrieval clients provide context and reader clients supply an answer policy that can only correct a label once the retrieval group has surfaced the gold evidence. A precedence knob \(g\in[0,1]\) tunes how strongly reader value is gated on retrieval (\(g{=}0\): no precedence; \(g{=}1\): strong precedence). Each rule scores clients on a validation card, selects a budget-feasible coalition, and is then served on a disjoint held-out test card; the served accuracy is the ground truth, and the same market and budget are used for both rules. Table~\ref{tab:ordered-serving-neutral} reports the outcome. Ordered valuation redistributes credit (ordered-vs-unordered Spearman \(0.725\), selected-set overlap \(0.690\)), but the paired held-out accuracy difference is statistically indistinguishable from zero across budgets at strong precedence (\(\Delta\in[-0.009,+0.011]\), all 95\% CIs covering zero, sign \(\approx\)~50/50). An exact-Shapley small market agrees (\(-0.011\pm0.058\)) and confirms the samplers recover their targets (sampled-ordered vs exact-ordered \(\rho{=}0.958\); sampled-unordered vs exact-symmetric \(\rho{=}0.935\)). We conclude that ordered valuation is \emph{serving-neutral}: it aligns payment with the causal serving order without changing task quality, which is a payment-fairness property rather than an accuracy improvement.

\begin{table}[t]
\centering
\caption{Ordered valuation is serving-neutral. Paired ordered\(-\)unordered held-out FEVER serving accuracy at strong precedence (\(g{=}1\)); every CI covers zero. Credit is nonetheless redistributed (bottom row).}
\label{tab:ordered-serving-neutral}
\scriptsize
\begin{tabular}{lrr}
\toprule
\rowcolor{headrow} Budget & \(\Delta\) accuracy (95\% CI) & \(P(\Delta>0)\) \\
\midrule
1.4 & \(+0.011\pm0.021\) & 0.47 \\
1.7 & \(-0.009\pm0.022\) & 0.40 \\
2.0 & \(+0.004\pm0.014\) & 0.42 \\
\midrule
\rowcolor{keyrow} Credit redistribution & Spearman 0.725 & overlap 0.690 \\
\bottomrule
\end{tabular}
\end{table}

\subsection{Computational Cost}

We decompose one market round into a one-time setup phase and a per-marginal sampling phase. Let \(n\) be the number of clients, \(m\) the number of sampled marginal pairs \((C,i)\), \(k\) the number of redundancy clusters, \(h\in[0,1]\) the direct-audit fraction, \(L\) the number of differentially private (DP) releases per round, and \(R\) the number of evidence records signed. The round cost is
\begin{align}
T(n,m,k,h) ={}& \underbrace{n\,T_{\mathrm{sk}} + n^2 T_{\mathrm{cl}} + n\,T_{\mathrm{reg}}}_{T_{\mathrm{init}}(n)} \nonumber\\
&{}+ m\!\cdot\!\big(2h\,T_{\mathrm{ev}}(n) + (1\!-\!h)\,T_{\mathrm{su}} + T_{\mathrm{fe}} + T_{\mathrm{up}}\big) \nonumber\\
&{}+ L\,T_{\mathrm{dp}} + R\,T_{\mathrm{sg}},
\label{eq:complexity}
\end{align}
where \(T_{\mathrm{sk}}\) is sketch construction per client (MinHash and centroids), \(T_{\mathrm{cl}}\) is pairwise redundancy clustering (linkage over sketch similarity, \(O(n^2)\) worst case), \(T_{\mathrm{reg}}\) is registry write plus provenance commit, \(T_{\mathrm{ev}}(n)\) is one coalition utility call (the only term that scales with \(n\) through retrieval index size or coalition payload), \(T_{\mathrm{su}}\) is one surrogate prediction, \(T_{\mathrm{fe}}\) is feature-vector construction, \(T_{\mathrm{up}}\) is the online surrogate update (online ridge: \(O(d^2)\) where \(d\) is feature dimension), \(T_{\mathrm{dp}}\) is one DP release (clip, noise, accountant update), and \(T_{\mathrm{sg}}\) is one signed-evidence write.
Utility evaluation dominates when \(n\), \(m\), or \(h\) is large; DP and signing scale with round-policy constants. A coarse fit to Table~\ref{tab:scale} gives
\[
T_{\mathrm{step}}(n) \approx 2.0\!\times\!10^{-3}\,n^{1.5}\,\mathrm{s},\quad
T_{\mathrm{init}}(n)\leq 1\,\mathrm{s},
\]
for \(n\geq16\). Table~\ref{tab:complexity-fit} reports measured versus fitted runtime; Table~\ref{tab:tev-microbench} separates fixed-shard and scale-sweep utility-call costs; Table~\ref{tab:msweep} shows that per-marginal evaluation dominates once \(m\) grows.

\begin{table}[t]
\centering
\caption{Complexity-model fit to Table~\ref{tab:scale}. Fitted constants are rounded; the scaling interpretation is intended for \(n\geq16\).}
\label{tab:complexity-fit}
\scriptsize
\begin{tabular}{rrrrr}
\toprule
\rowcolor{headrow} Clients \(n\) & Marginals \(m\) & Measured (s) & Fitted (s) & Rel.\ err. \\
\midrule
8 & 6 & 0.6495 & 0.444 & \(-31.6\%\) \\
8 & 16 & 1.1540 & 0.834 & \(-27.7\%\) \\
16 & 6 & 0.8888 & 0.821 & \(-7.6\%\) \\
16 & 16 & 1.5915 & 1.891 & \(+18.8\%\) \\
32 & 6 & 1.8874 & 1.916 & \(+1.5\%\) \\
32 & 16 & 5.0247 & 4.855 & \(-3.4\%\) \\
64 & 6 & 6.2499 & 4.974 & \(-20.4\%\) \\
64 & 16 & 15.1439 & 13.048 & \(-13.8\%\) \\
\bottomrule
\end{tabular}
\end{table}

\begin{table}[t]
\centering
\caption{Direct \(T_{\mathrm{ev}}(n)\) microbenchmark. Fixed-shard rows keep per-client shard size constant; scale-sweep rows use the same shard schedule as the scalability fit.}
\label{tab:tev-microbench}
\scriptsize
\begin{tabular}{llrr}
\toprule
\rowcolor{headrow} Shard schedule & Clients & Docs/client & ms/call \\
\midrule
Fixed smoke & 8 & 3 & 0.163 \\
Fixed smoke & 32 & 3 & 0.248 \\
Fixed smoke & 64 & 3 & 0.264 \\
Fixed smoke & 100 & 3 & 0.296 \\
Scale sweep & 8 & 22 & 1.371 \\
Scale sweep & 32 & 5 & 0.433 \\
Scale sweep & 64 & 3 & 0.265 \\
Scale sweep & 100 & 3 & 0.291 \\
\bottomrule
\end{tabular}
\end{table}

\begin{table}[t]
\centering
\caption{Projected runtime under the fitted model for \(m\in\{16,64,256,1024\}\). Calls counts assume \(h\!=\!1\) (full direct audits, two utility calls per marginal). Surrogate triage with \(h\!<\!1\) reduces total time proportionally.}
\label{tab:msweep}
\scriptsize
\begin{tabular}{rrrrr}
\toprule
\rowcolor{headrow} Clients \(n\) & Marginals \(m\) & Calls & \(T_{\mathrm{pred}}\) (s) & \(T_{\mathrm{step}}\) share \\
\midrule
64 & 16  & 32  & 13.05  & 99.0\% \\
64 & 64  & 128 & 51.80  & 99.7\% \\
64 & 256 & 512 & 206.83 & 99.9\% \\
64 & 1024 & 2048 & 826.92 & 99.98\% \\
\midrule
100 & 16  & 32  & 24.88  & 99.5\% \\
100 & 64  & 128 & 99.16  & 99.9\% \\
100 & 256 & 512 & 396.30 & 99.97\% \\
100 & 1024 & 2048 & 1584.83 & 99.99\% \\
\bottomrule
\end{tabular}
\end{table}

Combining the fits, \method{} valuation at fixed surrogate audit rate \(h\) costs
\[
T(n,m) = O\!\big(n^{1.5} hm\big) \;+\; O(n^2)
    \;+\; O\!\big(d^2 m\big) \;+\; O(L+R),
\]
where the first term (coalition utility evaluation) is dominant whenever \(hm\!\geq\!\Theta(n^{0.5})\) and where the \(O(n^2)\) clustering term is incurred once per round. Under the default direct-audit budget \(h\,m \leq O(n\log n)\), total cost is \(\tilde O(n^{2.5})\) per round, which is the relevant scaling for sizing the model-call budget against an operator latency SLO.

\subsection{Calibration, Redundancy, and Drift}

LCBs use calibrated intervals rather than raw sampled-marginal variance. In the current FEVER run, naive intervals cover 0.25 of leave-one-out values, while conformal calibration reaches 0.8875 with mean width 0.0408. Retrieval duplicate clustering uses provenance hashes plus
\[
s_{dup}(x,y)=0.65\,s_{tok}(x,y)+0.35\,s_{tri}(x,y),
\]
with adapter deployments adding delta-norm and target-module fingerprints. The hard-paraphrase recall at threshold 0.65 is 0.9575, so these checks are audit triggers, not complete guarantees. Other auxiliary results are consolidated as follows: collusion group caps suppress two- and three-client artificial complementarity in the tested suite; per-type LCB normalization improves adapter rare-client utility from 0.6625 to 0.6863; policy drift is visible when changing retriever or adapter-composition rules; and the software enclave adds only 0.226 ms per bounded signed release. Payments should not be reused across model, retriever, or policy changes.

\section{Payment, Budgeting, and Disputes}

\subsection{Payment Decomposition}

For each client, the ledger stores a decomposed payment:
\begin{align}
S_i &= \widehat{\phi}_i-\lambda\sigma_i-\beta c_i-\gamma\epsilon_i^{priv}
   -\eta\max(d_i,r_i)+\rho s_i, \\
p_i &= scale_B([S_i]_+),
\end{align}
where \(scale_B\) applies the budget-feasible scaling rule. The decomposition is important for review and dispute handling. A client should be able to see whether low payment is due to low estimated value, high uncertainty, duplicate risk, privacy restrictions, verified cost, or manipulation risk. The operator should be able to recompute the exact payment from the contract card and ledger fields.

\subsection{Budget Scaling}

If \(\sum_i \widetilde{p}_i>B\), the rule scales every positive raw payment by the same factor. This preserves payment order among clients with fixed adjustment terms and prevents the market from exceeding the budget. A deployment may add reserve checks after scaling: if a client's scaled payment falls below its verified participation cost, the client can be excluded and the remaining budget recomputed.

\subsection{Dispute Workflow}

A dispute proceeds in five steps. First, the client submits the disputed round id, artifact id, and reason. Second, the auditor retrieves the contract card, registry entry, value ledger, and evidence hashes. Third, the auditor reruns deterministic evaluations under the recorded seed where possible. Fourth, the auditor can request additional hidden probes or provenance review if duplicate, poison, or leakage flags are contested. Fifth, the ledger records a signed decision: unchanged, revalued, rejected, or escalated. The important design point is that the dispute is about auditable data artifacts, not an opaque payment number.

\subsection{Revaluation Policy}

Model, retriever, prompt, and safety-policy changes can alter \(U(C)\). \method{} treats a market round as immutable once settled: the original contract card, model hash, policy hash, utility weights, and evidence hashes define the value of that round. Later model or policy changes create a new round with a new contract card. The default policy is forward-looking revaluation rather than clawback, because clawbacks make participation risky and can punish honest clients for operator-side changes. Clawbacks are reserved for fraud cases where provenance, duplicate, or poison evidence shows that the original submission violated the signed contract. Figure~\ref{fig:multiround} illustrates six-round market dynamics with a slice shift after round three, showing why revaluation should be round-bound rather than silently rewriting old payments.

\begin{figure}[t]
\centering
\includegraphics[width=\linewidth]{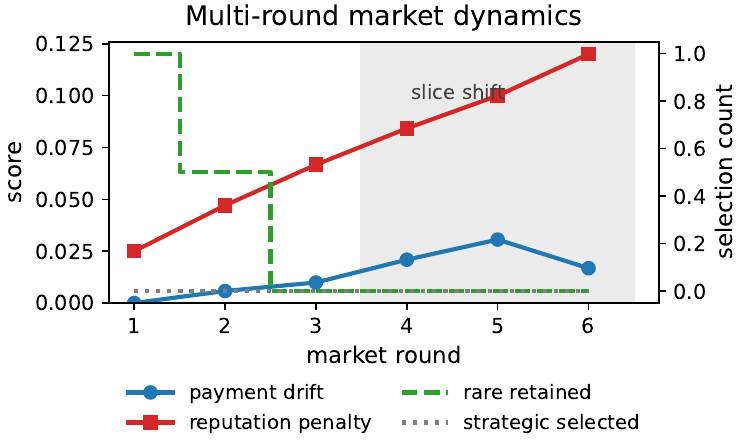}
\caption{Six-round market dynamics with client arrivals and a slice shift after round three, separating payment drift, reputation penalty, rare retention, and strategic selections.}
\label{fig:multiround}
\end{figure}

\subsection{Commit-Reveal and Operator Threats}

The client threat model is not sufficient by itself: an operator could collude with a seller by retrofitting a rare slice, changing hidden probes, or reweighting utility after seeing submissions. \method{} therefore uses a commit-reveal round protocol. Before the submission window opens, the operator publishes and signs a contract-card hash containing utility weights, rare-slice definitions, validation-slice commitments, audit policy, payment formula version, and dispute window. Clients submit artifacts against that hash. After the window closes, the operator reveals the card and evaluation seeds; auditors verify that the revealed card matches the pre-submission commitment. The multi-round stress test includes a post-hoc slice-shift condition; the ledger treats it as a new round, so prior payments remain bound to the original committed card rather than being silently recomputed.

\subsection{Operator and Auditor Adversaries}

The strongest auditability guarantees require more than an honest operator. If the operator selectively evaluates coalitions, the registry and ledger still expose which evaluation calls were made, but value completeness requires public sampling seeds, third-party reruns, or an external auditor. If the operator swaps the contract card, signed pre-window hashes and a public bulletin-board anchor detect the swap. If the operator redacts ledger entries, append-only storage and third-party hash anchoring reveal gaps but cannot reconstruct missing private payloads without client-side evidence. If the operator colludes with a seller, hidden-probe commitments, provenance hashes, and dispute replay limit post-hoc rare-slice and weight manipulation, but economic fairness depends on an independent auditor or buyer-side challenge process. If the auditor is compromised, signatures and public anchors preserve tamper evidence, but dispute decisions require auditor rotation, threshold signatures, or multi-auditor review. If keys are lost or stolen, old evidence remains verifiable only up to the last trusted key-rotation checkpoint. Table~\ref{tab:operator-threats} summarizes this operator/auditor threat surface, the required mitigation, and the residual risk in each case.

\begin{table*}[t]
\centering
\caption{Operator/auditor threat surface.}
\label{tab:operator-threats}
\scriptsize
\setlength{\tabcolsep}{2pt}
\begin{tabular}{@{}p{0.20\textwidth}p{0.25\textwidth}p{0.30\textwidth}p{0.16\textwidth}@{}}
\toprule
\rowcolor{headrow} Adversary action & What survives & Required mitigation & Residual risk \\
\midrule
Selective evaluation & Ledger shows sampled coalitions and missing calls & Public seeds, audit reruns, model-call counters & Private endpoints may be unavailable later \\
Contract-card swap & Pre-window signed hash detects mismatch & Public bulletin board or third-party hash anchor & If no anchor exists, clients rely on operator logs \\
Ledger redaction & Hash-chain gap is detectable & Append-only ledger, external checkpoints & Redacted payload evidence may need client replay \\
Seller/operator collusion & Committed weights and rare slices constrain retrofits & Hidden-probe commitments and buyer/auditor challenge & Collusion before card publication is governance risk \\
Auditor compromise & Signatures and anchors preserve raw evidence integrity & Auditor rotation, threshold review, public dispute summaries & Bad decisions can still delay payment \\
Key failure & Prior checkpoints remain verifiable & Key rotation, HSM/TEE-backed signing, revocation log & Evidence after compromise must be re-attested \\
\bottomrule
\end{tabular}
\end{table*}

\section{Threat Model}

The market assumes clients may be strategic but not omnipotent. They may split identities, copy or paraphrase artifacts, poison data, tune to public probes, inflate costs, exaggerate privacy constraints, or collude. The market assumes that clients cannot forge provenance commitments, cannot break cryptographic hashes, cannot observe hidden audit probes before submitting, and cannot force the operator to evaluate raw private payloads outside the declared interface. The operator is not assumed to be fully trusted for auditability claims. The base experiments assume the operator follows the published contract card, while the deployed protocol relies on signed cards, append-only ledgers, public hash anchors, attestation chains, and external dispute review to detect selective evaluation, card swaps, ledger redaction, and collusion. Table~\ref{tab:threat-model} summarizes the client-side threats, their market mitigations, and the residual risk of each.

\begin{table*}[t]
\centering
\caption{Threat model and mitigation summary.}
\label{tab:threat-model}
\scriptsize
\setlength{\tabcolsep}{2pt}
\begin{tabular}{@{}p{0.18\textwidth}p{0.26\textwidth}p{0.27\textwidth}p{0.20\textwidth}@{}}
\toprule
\rowcolor{headrow} Threat & Client action & Market mitigation & Residual risk \\
\midrule
Sybil split & Divide one portfolio across many identities & Redundancy clusters, group marginal discount, duplicate penalties & Semantic splits may evade weak clustering \\
Duplicate copy & Copy another client's passages or adapter & Provenance hashes, sketch similarity, duplicate discount & Paraphrases require stronger semantic checks \\
Poisoning & Improve public validation while harming hidden behavior & Hidden probes, risk score, audit penalty & Adaptive poisoners can target unknown policies \\
Prompt gaming & Tune prompts to public probes & Paraphrase probes, public-hidden gap, uncertainty discount & Probe leakage weakens defense \\
Privacy gaming & Claim restrictive privacy to reduce evaluation visibility & Lower confidence bound, alternative secure evaluation, reserve checks & Legitimate privacy may also widen intervals \\
Cost inflation & Overstate collection or compute cost & Verified cost, reserve prices, cost penalties & Verification may be domain-specific \\
Collusion & Split complementary artifacts across identities & Pairwise complementarity audit and group caps & True complementarity should not be over-penalized \\
Benchmark leakage & Submit data derived from validation answers & Provenance review and hidden evaluation slices & Perfect leakage detection is impossible \\
\bottomrule
\end{tabular}
\end{table*}

\section{Experimental Details}

\subsection{Baseline Definitions}

Equal payment assigns all clients the same score before budget selection. Volume ranks clients by artifact count or corpus size. Retrieval similarity ranks retrieval clients by similarity to validation queries or gold evidence sketches. Leave-one-client-out estimates \(U(N)-U(N\setminus\{i\})\), which is accurate but expensive and can behave differently from Shapley values when interactions are strong. FL-Shapley approximates marginal contribution over sampled client orderings. Shapley-UCB uses uncertainty-aware seller selection inspired by bandit upper-confidence bounds. \method{} differs by ranking with lower-confidence-bound value plus explicit market adjustments for cost, privacy, duplicate risk, manipulation risk, and scarcity.

\subsection{Experiment Construction Details}

FEVER RAG shards claim-evidence pairs into honest generalists, rare specialists, duplicates, sybils, and poisoners. The held-out serving tracks separate validation-card selection from final test-card scoring and report task metrics without using attack labels. The trained low-rank adapter track uses lightweight PyTorch adapters, while the HF/PEFT path uses \texttt{distilbert-base-uncased} to report real LoRA parameters, runtime, and coalition behavior. The PEFT track proves the interface and accounting path; large multi-adapter foundation-model valuation remains future scale-up work.

\begin{table*}[t]
\centering
\caption{Retrieve--read held-out FEVER serving proxy. Accuracy, macro-F1, evidence exact match (EM), and regret are computed only on the final test card and do not use attack labels. This table is supporting evidence; Table~\ref{tab:generator-backed-rag} is the primary served-pipeline result.}
\label{tab:heldout-serving}
\scriptsize
\setlength{\tabcolsep}{2pt}
\begin{tabular}{llccccrr}
\toprule
\rowcolor{headrow} Attack family & Method & Acc. & Macro-F1 & EM & Regret & Poison & Strategic \\
\midrule
Known flip & FL-Shapley & \(0.3685{\pm}0.0094\) & \(0.3136{\pm}0.0100\) & \(\mathbf{0.1343{\pm}0.0424}\) & \(0.0686{\pm}0.0399\) & 0.6 & 2.4 \\
Known flip & Leave-one-out & \(0.3428{\pm}0.0274\) & \(0.2978{\pm}0.0313\) & \(0.1257{\pm}0.0568\) & \(0.0943{\pm}0.0423\) & 0.6 & 2.4 \\
\rowcolor{keyrow} Known flip & \method{} & \(\mathbf{0.3914{\pm}0.0303}\) & \(\mathbf{0.3490{\pm}0.0498}\) & \(0.0914{\pm}0.0188\) & \(\mathbf{0.0457{\pm}0.0407}\) & \textbf{0.0} & \textbf{0.2} \\
\midrule
Hard paraphrase & FL-Shapley & \(0.3600{\pm}0.0399\) & \(0.3158{\pm}0.0521\) & \(0.0972{\pm}0.0321\) & \(0.0800{\pm}0.0563\) & 0.4 & 2.2 \\
Hard paraphrase & Leave-one-out & \(0.3457{\pm}0.0311\) & \(0.2999{\pm}0.0400\) & \(\mathbf{0.1171{\pm}0.0610}\) & \(0.0914{\pm}0.0368\) & 0.4 & 2.2 \\
\rowcolor{keyrow} Hard paraphrase & \method{} & \(\mathbf{0.3857{\pm}0.0137}\) & \(\mathbf{0.3463{\pm}0.0338}\) & \(0.0828{\pm}0.0243\) & \(\mathbf{0.0543{\pm}0.0340}\) & \textbf{0.0} & \textbf{0.2} \\
\midrule
Adaptive keyword & FL-Shapley & \(0.3714{\pm}0.0237\) & \(0.3252{\pm}0.0420\) & \(0.1029{\pm}0.0375\) & \(0.0657{\pm}0.0431\) & 0.4 & 2.4 \\
Adaptive keyword & Leave-one-out & \(0.3514{\pm}0.0332\) & \(0.3029{\pm}0.0388\) & \(\mathbf{0.1457{\pm}0.0422}\) & \(0.0857{\pm}0.0380\) & 0.8 & 2.2 \\
\rowcolor{keyrow} Adaptive keyword & \method{} & \(\mathbf{0.3857{\pm}0.0317}\) & \(\mathbf{0.3437{\pm}0.0495}\) & \(0.0886{\pm}0.0216\) & \(\mathbf{0.0514{\pm}0.0424}\) & \textbf{0.0} & \textbf{0.2} \\
\bottomrule
\end{tabular}
\end{table*}

\subsection{Hyperparameters and Reporting Conventions}

We use fixed random seeds for each experiment. Budgeted selection uses the same budget within each track for all methods. Risk-adjusted utility uses the same strategic penalty and rare-client bonus within a track. Harmful-client AUC is reported only when harmful labels exist. PEFT runtime fields are reported only when the optional PEFT dependencies and cached \texttt{distilbert-base-uncased} model are available.


\begin{thebibliography}{44}
\providecommand{\url}[1]{#1}
\csname url@samestyle\endcsname
\providecommand{\newblock}{\relax}
\providecommand{\bibinfo}[2]{#2}
\providecommand{\BIBentrySTDinterwordspacing}{\spaceskip=0pt\relax}
\providecommand{\BIBentryALTinterwordstretchfactor}{4}
\providecommand{\BIBentryALTinterwordspacing}{\spaceskip=\fontdimen2\font plus
\BIBentryALTinterwordstretchfactor\fontdimen3\font minus
  \fontdimen4\font\relax}
\providecommand{\BIBforeignlanguage}[2]{{%
\expandafter\ifx\csname l@#1\endcsname\relax
\typeout{** WARNING: IEEEtran.bst: No hyphenation pattern has been}%
\typeout{** loaded for the language `#1'. Using the pattern for}%
\typeout{** the default language instead.}%
\else
\language=\csname l@#1\endcsname
\fi
#2}}
\providecommand{\BIBdecl}{\relax}
\BIBdecl

\bibitem{kairouz2021advancesfl}
P.~Kairouz, H.~B. McMahan, B.~Avent, A.~Bellet, M.~Bennis, A.~N. Bhagoji,
  K.~Bonawitz, Z.~Charles, G.~Cormode, R.~Cummings \emph{et~al.}, ``Advances
  and open problems in federated learning,'' \emph{Foundations and Trends in
  Machine Learning}, vol.~14, no. 1--2, pp. 1--210, 2021.

\bibitem{yang2023flshapleypareto}
X.~Yang, S.~Xiang, C.~Peng \emph{et~al.}, ``Federated learning incentive
  mechanism design via shapley value and pareto optimality,'' \emph{Axioms},
  vol.~12, no.~7, p. 636, 2023.

\bibitem{ghorbani2019datashapley}
A.~Ghorbani and J.~Zou, ``Data shapley: Equitable valuation of data for machine
  learning,'' in \emph{Proceedings of the 36th International Conference on
  Machine Learning}, 2019, pp. 2242--2251.

\bibitem{jia2019efficient}
R.~Jia, D.~Dao, B.~Wang, F.~A. Hubis, N.~Hynes, N.~M. G{\"u}rel, B.~Li,
  C.~Zhang, D.~Song, and C.~J. Spanos, ``Efficient task-specific data valuation
  for nearest neighbor algorithms,'' \emph{Proceedings of the VLDB Endowment},
  vol.~12, no.~11, pp. 1610--1623, 2019.

\bibitem{fan2025fedfmchallenges}
T.~Fan, H.~Gu, X.~Cao \emph{et~al.}, ``Ten challenging problems in federated
  foundation models,'' \emph{IEEE Transactions on Knowledge and Data
  Engineering}, vol.~37, no.~7, pp. 4314--4337, 2025.

\bibitem{shapley1953value}
L.~S. Shapley, ``A value for n-person games,'' in \emph{Contributions to the
  Theory of Games II}.\hskip 1em plus 0.5em minus 0.4em\relax Princeton
  University Press, 1953, pp. 307--317.

\bibitem{zhang2024shapleyucb}
K.~Chen and Z.~Xu, ``Federated learning for data market: Shapley-ucb for seller
  selection and incentives,'' \emph{arXiv preprint arXiv:2410.09107}, 2024.

\bibitem{wang2024fedave}
Z.~Wang \emph{et~al.}, ``Fedave: Adaptive data value evaluation framework for
  collaborative fairness in federated learning,'' \emph{Neurocomputing}, vol.
  574, p. 127227, 2024.

\bibitem{zhang2025fmincentive}
N.~Zhang, X.~Xu, X.~Liu, J.~Wu, and H.~Tang, ``Incentive mechanism of
  foundation model enabled cross-silo federated learning,'' \emph{Scientific
  Reports}, vol.~15, p. 24181, 2025.

\bibitem{vickrey1961counterspeculation}
W.~Vickrey, ``Counterspeculation, auctions, and competitive sealed tenders,''
  \emph{Journal of Finance}, vol.~16, no.~1, pp. 8--37, 1961.

\bibitem{clarke1971multipart}
E.~H. Clarke, ``Multipart pricing of public goods,'' \emph{Public Choice},
  vol.~11, pp. 17--33, 1971.

\bibitem{groves1973incentives}
T.~Groves, ``Incentives in teams,'' \emph{Econometrica}, vol.~41, no.~4, pp.
  617--631, 1973.

\bibitem{miller2005peer}
N.~Miller, P.~Resnick, and R.~Zeckhauser, ``Eliciting informative feedback: The
  peer-prediction method,'' in \emph{Management Science}, vol.~51, no.~9, 2005,
  pp. 1359--1373.

\bibitem{jiang2021flmarket}
Z.~Jiang, Y.~Cao, Y.~Wang, H.~Chen, and C.~Xu, ``{FL-Market}: Trading private
  models in federated learning,'' \emph{arXiv preprint arXiv:2106.04384}, 2021.

\bibitem{fung2018foolsgold}
C.~Fung, C.~J.~M. Yoon, and I.~Beschastnikh, ``Mitigating sybils in federated
  learning poisoning,'' \emph{arXiv preprint arXiv:1808.04866}, 2018.

\bibitem{koh2017influence}
P.~W. Koh and P.~Liang, ``Understanding black-box predictions via influence
  functions,'' in \emph{Proceedings of the 34th International Conference on
  Machine Learning}, 2017, pp. 1885--1894.

\bibitem{yoon2020dvrl}
J.~Yoon, S.~O. Arik, and T.~Pfister, ``Data valuation using reinforcement
  learning,'' in \emph{Proceedings of the 37th International Conference on
  Machine Learning (ICML)}, ser. Proceedings of Machine Learning Research, vol.
  119, 2020, pp. 10\,842--10\,851.

\bibitem{kwon2022betashapley}
Y.~Kwon and J.~Zou, ``Beta shapley: A unified and noise-reduced data valuation
  framework for machine learning,'' in \emph{International Conference on
  Artificial Intelligence and Statistics}, 2022, pp. 8780--8802.

\bibitem{ye2025documentvaluation}
Z.~Ye and H.~Yoganarasimhan, ``Fair document valuation in llm summaries via
  shapley values,'' \emph{arXiv preprint arXiv:2505.23842}, 2025.

\bibitem{han2025ragauctions}
M.~Han, S.~A. Esmaeili, M.~Albert, and H.~Xu, ``Data auctions for retrieval
  augmented generation,'' \emph{arXiv preprint arXiv:2508.16007}, 2025.

\bibitem{lewis2020rag}
P.~Lewis, E.~Perez, A.~Piktus, F.~Petroni, V.~Karpukhin, N.~Goyal,
  H.~K{\"u}ttler, M.~Lewis, W.-t. Yih, T.~Rockt{\"a}schel \emph{et~al.},
  ``Retrieval-augmented generation for knowledge-intensive nlp tasks,'' in
  \emph{Advances in Neural Information Processing Systems}, 2020.

\bibitem{hu2022lora}
E.~J. Hu, Y.~Shen, P.~Wallis, Z.~Allen-Zhu, Y.~Li, S.~Wang, L.~Wang, and
  W.~Chen, ``Lora: Low-rank adaptation of large language models,'' in
  \emph{International Conference on Learning Representations}, 2022.

\bibitem{ouyang2022training}
L.~Ouyang, J.~Wu, X.~Jiang, D.~Almeida, C.~L. Wainwright, P.~Mishkin, C.~Zhang,
  S.~Agarwal, K.~Slama, A.~Ray \emph{et~al.}, ``Training language models to
  follow instructions with human feedback,'' in \emph{Advances in Neural
  Information Processing Systems}, 2022.

\bibitem{cho2024hetlora}
Y.~J. Cho, L.~Liu, Z.~Xu, A.~Fahrezi, and G.~Joshi, ``Heterogeneous lora for
  federated fine-tuning of on-device foundation models,'' \emph{arXiv preprint
  arXiv:2401.06432}, 2024.

\bibitem{chen2024rolora}
S.~Chen, Y.~Ju, H.~Dalal, Z.~Zhu, and A.~Khisti, ``Robust federated finetuning
  of foundation models via alternating minimization of lora,'' \emph{arXiv
  preprint arXiv:2409.02346}, 2024.

\bibitem{singhal2025fedexlora}
R.~Singhal, K.~Ponkshe, and P.~Vepakomma, ``{FedEx-LoRA}: Exact aggregation for
  federated and efficient fine-tuning of foundation models,'' \emph{arXiv
  preprint arXiv:2410.09432}, 2025.

\bibitem{bian2024lorafair}
J.~Bian, L.~Wang, L.~Zhang, and J.~Xu, ``{LoRA-FAIR}: Federated {LoRA}
  fine-tuning with aggregation and initialization refinement,'' \emph{arXiv
  preprint arXiv:2411.14961}, 2024.

\bibitem{wang2024flora}
Y.~Wang \emph{et~al.}, ``{FLoRA}: Federated fine-tuning large language models
  with heterogeneous low-rank adaptations,'' in \emph{Advances in Neural
  Information Processing Systems (NeurIPS)}, 2024.

\bibitem{huang2026sfedlora}
J.~Huang, X.~Wu, T.~He, and Q.~Lao, ``Stabilized fine-tuning with lora in
  federated learning: Mitigating the side effect of client size and rank via
  the scaling factor,'' \emph{arXiv preprint arXiv:2603.08058}, 2026.

\bibitem{kou2026winflora}
M.~Kou, X.~Xia, Z.~Wang, I.~Khalil, R.~Luo, J.~Zhou, and M.~Xue, ``{WinFLoRA}:
  Incentivizing client-adaptive aggregation in federated {LoRA} under privacy
  heterogeneity,'' in \emph{Proceedings of the ACM Web Conference 2026}, 2026,
  pp. 5241--5252.

\bibitem{muhamed2025corag}
A.~Muhamed, M.~Diab, and V.~Smith, ``{CoRAG}: Collaborative retrieval-augmented
  generation,'' in \emph{Proceedings of the 2025 Conference of the Nations of
  the Americas Chapter of the Association for Computational Linguistics
  (NAACL): Short Papers}, 2025, pp. 265--276.

\bibitem{fede4rag2025}
Q.~Mao \emph{et~al.}, ``{FedE4RAG}: Privacy-preserving federated embedding
  learning for retrieval-augmented generation,'' \emph{arXiv preprint
  arXiv:2504.19101}, 2025.

\bibitem{lu2024datameasurements}
C.~Lu, M.~M. Amiri, and R.~Raskar, ``Data measurements for decentralized data
  markets,'' \emph{arXiv preprint arXiv:2406.04257}, 2024.

\bibitem{zheng2024ads}
X.~Zheng, X.~Chang, R.~Jia, and Y.~Tan, ``Towards data valuation via asymmetric
  data shapley,'' \emph{arXiv preprint arXiv:2411.00388}, 2024.

\bibitem{zheng2025ads}
X.~Zheng, Y.~Huang, X.~Chang, R.~Jia, and Y.~Tan, ``Rethinking data value:
  Asymmetric data shapley for structure-aware valuation in data markets and
  machine learning pipelines,'' \emph{arXiv preprint arXiv:2511.12863}, 2025.

\bibitem{chi2024pcwinter}
H.~Chi, Z.~Yang, L.~Zeng, W.~Fan, and Y.~Ma, ``Precedence-constrained winter
  value for effective graph data valuation,'' \emph{arXiv preprint
  arXiv:2402.01943}, 2024.

\bibitem{li2025fedowen}
H.~KhademSohi, H.~Hemmati, J.~Zhou, and S.~Drew, ``Owen sampling accelerates
  contribution estimation in federated learning,'' \emph{arXiv preprint
  arXiv:2508.21261}, 2025.

\bibitem{entropy2026datafree}
A.~Ukaye, M.~Abdu-Aguye, N.~Tastan, and K.~Nandakumar, ``Data-free contribution
  estimation in federated learning using gradient von neumann entropy,''
  \emph{arXiv preprint arXiv:2604.22562}, 2026.

\bibitem{addison2024cfedrag}
P.~Addison, M.-T.~H. Nguyen, T.~Medan, J.~Shah, M.~T. Manzari, B.~McElrone,
  L.~Lalwani, A.~More, S.~Sharma, H.~R. Roth, I.~Yang, C.~Chen, D.~Xu,
  Y.~Cheng, A.~Feng, and Z.~Xu, ``{C-FedRAG}: A confidential federated
  retrieval-augmented generation system,'' \emph{arXiv preprint
  arXiv:2412.13163}, 2024.

\bibitem{yan2026fweb3}
P.~Yan, S.~Liang, Y.~Hua, L.~Jiang, K.~Yu, Y.~Sun, Y.~Zhang, T.~Song, N.~Hu,
  X.~Liang, B.~He, and H.~Guan, ``Fweb3: A practical incentive-aware federated
  learning framework,'' \emph{arXiv preprint arXiv:2603.00666}, 2026.

\bibitem{faigle1992precedence}
U.~Faigle and W.~Kern, ``The {S}hapley value for cooperative games under
  precedence constraints,'' \emph{International Journal of Game Theory},
  vol.~21, no.~3, pp. 249--266, 1992.

\bibitem{owen1977values}
G.~Owen, ``Values of games with a priori unions,'' in \emph{Mathematical
  Economics and Game Theory}.\hskip 1em plus 0.5em minus 0.4em\relax Springer,
  1977, pp. 76--88.

\bibitem{bun2016concentrated}
M.~Bun and T.~Steinke, ``Concentrated differential privacy: Simplifications,
  extensions, and lower bounds,'' in \emph{Theory of Cryptography Conference},
  2016, pp. 635--658.

\bibitem{mironov2017renyi}
I.~Mironov, ``R\'enyi differential privacy,'' in \emph{2017 IEEE 30th Computer
  Security Foundations Symposium (CSF)}.\hskip 1em plus 0.5em minus 0.4em\relax
  IEEE, 2017, pp. 263--275.

\end{thebibliography}
\end{document}